\documentclass[12pt,a4paper,reqno,oneside,nosumlimits]{amsart}
\usepackage{amsmath,amssymb,amsthm,eucal,url}

\usepackage{cite} 

\allowdisplaybreaks[4] 
\newenvironment{nouppercase}{\renewcommand{\uppercasenonmath}[1]{}}{} 

\theoremstyle{plain}
\newtheorem{theorem}{Theorem}
\newtheorem{lemma}[theorem]{Lemma}
\newtheorem{prop}[theorem]{Proposition}

\newcommand{\RR}{\mathbb{R}}
\newcommand{\CC}{\mathbb{C}}
\newcommand{\NN}{\mathbb{N}}
\newcommand{\Cl}{\mathcal{C}}

\newcommand{\Dsl}{{\ooalign{\(D\)\cr\hidewidth\(/\)\hidewidth\cr}}} 
\newcommand{\dnabla}{{\ooalign{\(\nabla\)\cr\hidewidth\(/\)\hidewidth\cr}}} 

\newcommand{\cB}{\mathcal{B}}

\newcommand{\cH}{\mathcal{H}}
\newcommand{\cG}{\mathcal{G}}

\newcommand{\cO}{\mathcal{O}}
\newcommand{\cP}{\mathcal{P}}
\newcommand{\cQ}{\mathcal{Q}}

\newcommand{\dd}{\,\mathrm{d}}
\newcommand{\rmi}{\mathrm{i}}
\newcommand{\dom}{\mathop{\mathcal{D}}}
\newcommand{\qdom}{\mathop{\mathcal{Q}}}

\DeclareMathOperator{\tr}{tr}
\DeclareMathOperator{\vspan}{span}
\DeclareMathOperator{\dist}{dist}
\DeclareMathOperator{\ddiv}{div}

\begin{document}

\title{\Large Dirac operators on hypersurfaces\\ as large mass limits}

\author[A. Moroianu]{\large Andrei Moroianu}
\address[A. Moroianu]{Laboratoire de Math\'ematiques d'Orsay, Univ.~Paris-Sud, CNRS, Universit\'e Paris-Saclay, 91405 Orsay, France}
\email{andrei.moroianu@math.cnrs.fr}
\urladdr{http://moroianu.perso.math.cnrs.fr}

\author[T. Ourmi\`eres-Bonafos]{\large Thomas Ourmi\`eres-Bonafos}
\address[T. Ourmi\`eres-Bonafos]{CNRS \& Universit\'e Paris-Dauphine, PSL  University, CEREMADE, Place de Lattre de Tassigny, 75016 Paris, France}
\email{ourmieres-bonafos@ceremade.dauphine.fr}
\urladdr{https://www.ceremade.dauphine.fr/~ourmieres/}

\author[K. Pankrashkin]{\large Konstantin Pankrashkin}
\address[K. Pankrashkin]{Laboratoire de Math\'ematiques d'Orsay, Univ.~Paris-Sud, CNRS, Universit\'e Paris-Saclay, 91405 Orsay, France}
\email{konstantin.pankrashkin@math.u-psud.fr}
\urladdr{http://www.math.u-psud.fr/~pankrashkin/}

\begin{abstract}
We show that the eigenvalues of the intrinsic Dirac operator on the boundary of a Euclidean domain can be obtained
as the limits of eigenvalues of Euclidean Dirac operators, either in the domain with a MIT-bag type boundary condition
or in the whole space, with a suitably chosen zero order mass term.
\end{abstract}

\subjclass[2010]{81Q05, 53C80, 35P15, 58C40}

\begin{nouppercase} 
\maketitle
\end{nouppercase} 

{\small

\tableofcontents

}

\section{Introduction}

\subsection{Problem setting and main results}

The aim of the present paper is to make a new link between a number of recent papers on Dirac operators in bounded Euclidean domains with the theory of Dirac operators on manifolds, which is a classical topic in Riemannian geometry. Namely, let $\Omega\subset\RR^n$ be a bounded domain with smooth boundary $\Sigma$.  We are going to show that
the intrinsic Dirac operator $\Dsl$, which acts on sections of the spinor bundle of $\Sigma$,
can be interpreted as a limit of Euclidean Dirac operators, either in $\Omega$ with a suitable boundary
condition, or in the whole of $\RR^n$, with a suitably chosen term containing a large mass.

For $n\ge 2$ and $N:=2^{[\frac{n+1}{2}]}$ let $\alpha_1,\dots,\alpha_{n+1}$ be anticommuting Hermitian $N\times N$ matrices with $\alpha_j^2=I_N$, where
$I_N$ is the $N\times N$ identity matrix.
The associated Dirac operator with a mass $m\in\RR$ acts on functions $u:\RR^n\to \CC^N$ (spinors) by the differential expression
\begin{equation}
 \label{eqedm}
D_m u=-\rmi\sum_{j=1}^n \alpha_j \dfrac{\partial u}{\partial x_j} + m \alpha_{n+1} u,
\end{equation}
see e.g. \cite{thaller}. We remark that the expression $D_m$ does not correspond to the intrinsic Dirac operator on $\RR^n$ (see Subsection~\ref{quad2}) and can be interpreted
as follows: the intrinsic operator $\widetilde D$ in $\RR^{n+1}$ is defined as
\[
\widetilde D v=-\rmi\sum_{j=1}^{n+1} \alpha_j \dfrac{\partial v}{\partial x_j},
\]
and acts on functions $v:\RR^{n+1}\to \CC^N$, then assuming that $v$ is of the form $v(x_1,\dots,x_{n+1})=e^{\rmi m x_{n+1}} u(x_1,\dots,x_n)$ one obtains
$\widetilde D v=e^{\rmi m x_{n+1}} D_m u$.

For $x=(x_1,\dots,x_n)\in\RR^n$ we define the associated $N\times N$ matrices $\Gamma(x)$ by
\begin{equation}
   \label{ekg1}
\Gamma(x):=\sum_{j=1}^n x_j\alpha_j.
\end{equation}
Denote by $\nu$ the unit normal at $\Sigma$ pointing to the exterior of~$\Omega$ and consider the $N\times N$ matrices
\begin{equation}
   \label{ekg3}
\cB(s):=-\rmi \alpha_{n+1} \Gamma\big(\nu(s)\big), \quad s\in\Sigma.
\end{equation}
By the Dirac operator $A_m$ in $\Omega$ with a mass $m\in\RR$ and the infinite mass boundary condition (also called MIT Bag boundary condition)
we mean the operator in $L^2(\Omega,\CC^N)$ given by
\[
A_m u=D_m u
\]
on the domain $\dom(A_m)=\big\{u\in H^1(\Omega,\CC^N): \, u=\cB u \text{ on } \Sigma\big\}$,
which is self-adjoint with compact resolvent (see Subsection~\ref{quad1}). In addition, for $m,M\in\RR$ we consider the following operator $B_{m,M}$ in $L^2(\RR^n,\CC^N)$,
which is the Dirac operator in the whole space with the mass $m$ in $\Omega$ and the mass $M$ outside~$\Omega$, i.e.
\begin{gather*}
B_{m,M}=D_0  +  \big[m 1_{\Omega} + M(1-1_\Omega)\big]\alpha_{n+1} \equiv D_m+(M-m)(1-1_\Omega)\,\alpha_{n+1}
\end{gather*}
with domain $\dom(B_{m,M})=H^1(\RR^n,\CC^N)$.
We are going to show that the eigenvalues of the intrinsic Dirac operator $\Dsl$ (whose construction is briefly reviewed in Subsection~\ref{quad2})
and of the Euclidean Dirac operators $A_m$ and $B_{m,M}$,
are related to each other  for suitable values of $m$ and $M$.

For a self-adjoint lower semibounded operator $T$ and $j\in\NN$ we denote by $E_j(T)$
the $j$th eigenvalue of $T$, if it exists,
when enumerated in the non-decreasing order and counted with multiplicities.
First we show that the eigenvalues of $\Dsl^2$ on $\Sigma$
are the limits of the eigenvalues of the square of the MIT Bag Dirac operator $A_m$ on $\Omega$
for large negative $m$:
\begin{theorem}\label{thm1a}
For each $j\in\NN$ there holds $E_j(\Dsl^2)=\lim_{m\to-\infty}E_j(A_m^2)$.
\end{theorem}

Then we show that, in turn, for any fixed $m$, the MIT Bag Dirac operators $A_m$ on $\Omega$ can be viewed as
the limits of the Dirac operators $B_{m,M}$ in the whole space with a large mass outside $\Omega$ (which justifies
the use of the term ``infinite mass boundary condition''):

\begin{theorem}\label{thm2}
For each $j\in\NN$  and $m\in\RR$ there holds $E_j(A_m^2)=\lim_{M\to+\infty} E_j(B_{m,M}^2)$.
\end{theorem}

Finally, by an additional construction we find an asymptotic regime in which the eigenvalues of $\Dsl^2$ on $\Sigma$
are directly recovered as the limits of the eigenvalues of the square of the Dirac operator $B_{m,M}^2$ on the whole space:

\begin{theorem}\label{thm3}
For each $j\in\NN$ the eigenvalue $E_j(B_{m,M}^2)$ converges
to $E_j(\Dsl^2)$ as $m\to -\infty$ and $M\to+\infty$ with $m/M\to 0$.
\end{theorem}

Let us comment on the three theorems. In the recent paper \cite{ALTR16} the operator $A_m$ in three dimensions was considered,
and it was shown that for each $j\in\NN$ one has $\lim_{m\to-\infty}E_j(A_m^2)=E_j(L)$ for some operator $L$ on $\Sigma$ given by its sesquilinear form. Hence, this result is extended in two directions: first, we consider arbitrary dimensions and, second, we show that the operator $L$ in question is in fact unitarily equivalent to $\Dsl^2$, which is our main observation.
Some analogs of Theorem~\ref{thm2} in two and three dimensions were obtained very recently in~\cite{ALTR18,BCLTS18,SW},
and we extend them to all dimensions. The result of Theorem~\ref{thm3} providing an interpretation of $\Dsl$ using an infinite
mass jump on $\Sigma$ does not seem to have previous analogs. In a sense, it can be viewed as a potential-induced
collapse by analogy with Dirac operators on manifolds converging to a lower-dimensional structure~\cite{lott,roos}.
As a possible application of our results, we remark that estimating the central gap (i.e. the first eigenvalue) of $A_m$ or $B_{m,M}$
in the respective asymptotic regime is reduced to the eigenvalue estimate for the Dirac operator $\Dsl$, for which a number of results are available: we refer to the book \cite{ginoux} for a review.

The text is organized as follows. In Subsection~\ref{sec-not} we recall a link between self-adjoint operators
and sesquilinear forms, choose a suitable notation, and then recall two important tools of the spectral analysis:
the min-max characterization of the eigenvalues and the monotone convergence.
In Section~\ref{sec-quad} we construct the sesquilinear forms for the squares of all the Dirac operators in question, which will
allow one to obtain eigenvalue estimates based on the min-max principle:
in Subsection~\ref{quad1} we recall the definition of various curvatures of $\Sigma$
and study $A_m$ and $B_{m,M}$, and in Subsection~\ref{quad2} we introduce an operator $L$,
which already appeared in \cite{ALTR16} for the three-dimensional case, and
prove that it is unitary equivalent to $\Dsl^2$. The unitary equivalence 
is shown using a Schr\"odinger-Lichnerowicz formula for extrinsic Dirac operators whose elementary proof for our Euclidean
setting is given in Appendix~\ref{sec-lichn} for reader's convenience. In Section~\ref{sec-prelim} we collect
some preliminary constructions: in Subsection~\ref{ssec1d} we study the eigenvalues and the eigenfunctions
of one-dimensional Laplacians $S$ and $S'$ with a large parameter in the boundary conditions, and in Subsection~\ref{sec-curv}
we give some computations in tubular coordinates near $\Sigma$.

In Section~\ref{sec-thm1} we prove Theorem~\ref{thm1a}. We first reduce the problem to the spectral analysis
is small tubular $\delta$-neighborhoods of $\Sigma$, and in order to work in $\Sigma\times(0,\delta)$ we use the computations
from Subsection~\ref{sec-curv}. The upper bound is obtained by taking as test functions the tensor products
of the eigenfunctions of (a small perturbation of) the effective operator $L$ on $\Sigma$ with the first eigenfunction
of the model operator $S$ in the normal direction.
For the lower bound we perform a unitary transform, which is just the expansion in eigenfunctions of the second model operator $S'$ in the normal variable,
thus transforming the problem into the study of a monotonically increasing sequence of operators. A simple application of the respective
machinery presented in Subsection~\ref{sec-not} then shows that only the projection onto the lowest eigenfunction of $S'$ contributes
to the asymptotics of the individual eigenvalues, which induces an effective operator acting on $\Sigma$ only.

The proof of Theorem~\ref{thm2} is presented in Section~\ref{sec-thm2}. To establish the upper bound we construct first an extension operator from $\Sigma$ to the exterior of $\Omega$
with a suitable control in terms of the mass $M$, and then use the corresponding extensions of the eigenfunctions of $A_m$ to construct test functions for $B_{m,M}$
used in the min-max principle.
For the lower bound we first decouple the two sides of $\Omega$ in order to deal separately with $\Omega$ and its exterior, then it is easily seen that the exterior
does not contribute to the lowest eigenvalues, while the part in $\Omega$ appears to be monotonically increasing in $M$ and then easily handled with the help of the monotone convergence.
The overall scheme here is very close to the one used in~\cite{SW} for the two-dimensional case.

In Section~\ref{sec-thm3} we prove Theorem~\ref{thm3}. The proof is essentially by combining in a new way various components from the preceding analysis,
but we still provide a complete self-contained argument. The upper bound is obtained  by taking the eigenfunctions of the operator $L$ on $\Sigma$
and extending them on both sides of $\Sigma$ by taking tensor products
with the first eigenfunctions of the model operators $S$ and $S'$ in the two normal directions, and then using them as test functions
in the min-max principle for $B_{m,M}^2$. For the lower bound we again decouple the two sides of $\Sigma$ and eliminate the exterior of $\Omega$
as in Theorem~\ref{thm2}. The analysis of the part in $\Omega$ is then quite similar to the one in Theorem~\ref{thm1a}: one is first reduced to the analysis in a thin tubular neighborhood
of $\Omega$, and then one applies a unitary transform in order obtain a monotone family with an explicit limit operator.
As will be seen from the proof, the domain $\Omega$ and its exterior play symmetric roles,
and, as a result, the eigenvalue convergence in Theorem~\ref{thm3} also holds in the asymptotic regime $m\to+\infty$, $M\to -\infty$, $M/m\to 0$.

Our approach based on the monotone convergence was chosen on purpose in order to obtain the main terms in a transparent way
and to be able to concentrate on the geometric aspects. A more precise analysis involving remainder estimates
and a more detailed operator convergence should be possible in the spirit
of the recent works on specific dimensions, e.g.~\cite{ALTR18,ALTR16,HOBP},
but a rigorous implementation requires a considerably higher technical effort,
and we prefer to discuss the related aspects in a separate forthcoming paper.

\subsection{Notation, min-max principle, monotone convergence}\label{sec-not}

The most part of the subsequent spectral analysis is based on the min-max principle for the eigenvalues of self-adjoint operators
and uses rather sesquilinear forms than operators (in particular, most operators are introduced just through their sesquilinear forms,
while the action and the domain of the operators are not specified explicitly). In order to avoid potential confusions, and to make the presentation
more accessible to non-experts, we recall here
some basic facts of the theory and introduce some notation.

Let $\cG$ be a Hilbert space, then by $\langle \cdot,\cdot\rangle_\cG$ we denote the scalar product in $\cG$, which is assumed antilinear with respect to the \emph{first} argument,
and the associated norm is denoted $\|\cdot\|_\cG$.

A sesquilinear form $t$ in $\cG$ defined on a subspace $\dom(t)$ of $\cG$ is a map
\[
\dom(t)\times\dom(t)\ni(u,v)\mapsto t(u,v)\in \CC
\]
which is antilinear with respect to the first argument and linear with respect to the second one, and it is called Hermitian
if $t(v,u)=\overline{t(u,v)}$ for all $u,v\in\dom(t)$. As a consequence of the polar identity, a Hermitian sesquilinear form $t$
is uniquely determined by its diagonal values $t(u,u)$ with $u\in\dom(t)$.
An Hermitian sesquilinear form $t$ is called lower semibounded
if there is $c\in\RR$ such that $t(u,u)\ge c\|u\|^2_\cG$ for all $u\in\dom(t)$.
Such a form is then called closed if $\dom(t)$ endowed with the scalar product
$\langle u,v\rangle_t:=t(u,v)+(1-c)\langle u,v\rangle_\cG$
is a Hilbert space. With such a sesquilinear form $t$
one associates a self-adjoint operator $T$ in $\cG$ uniquely defined
by the following two conditions:
(a) the domain $\dom(T)$ of $T$ is contained in $\dom(t)$ and (b) $t(u,v)=\langle u, Tv\rangle_\cG$ all $u,v\in\dom(T)$,
and we then say that $T$ is the \emph{self-adjoint operator generated by the form~$t$}.
It is worth noting that $\dom(T)\ne \dom(t)$ in general.

On the other hand, let $T$ be a self-adjoint operator in $\cG$ with  domain $\dom(T)$. It is called lower semibounded
if for some $c\in \RR$ one has $\langle u,T u\rangle_\cG\ge c\|u\|^2_\cG$ for all $u\in\dom(T)$, or $T\ge c$ for short.
In such a case, the completion of $\dom(T)$ with respect to the scalar product $\langle u,v\rangle_Q:=\langle u, Tv\rangle_\cG +(1-c)\langle u,v\rangle_\cG$
is called the \emph{form domain} of $T$ and is denoted by $\qdom(T)$. The map $\dom(T)\times\dom(T)\ni(u,v)\mapsto \langle u, Tv\rangle_\cG$
then uniquely extends to a closed lower semibounded Hermitian sesquilinear form $t$
with domain $\dom(t)=\qdom(T)$, which will be called the \emph{sesquilinear form generated by the operator $T$}.
In turn, $T$ is exactly the self-adjoint operator generated by this form~$t$. To have a shorter writing (and to reduce the number of symbols in use),
we will write
\[
T[u,v]:=t(u,v) \text{ for } u,v\in\qdom(T),
\]
in particular, one has the simple equality $T[u,v]=\langle u, T v\rangle_\cG$ if $v\in\dom(T)$.
We further recall that due to the spectral theorem we have
\begin{gather*}
\qdom(T)=\dom\big({\sqrt{T-c}}\big)=\dom(\sqrt{|T|}),\\
T[u,v]\equiv t(u,v)=\langle \sqrt{T-c}\, u,\sqrt{T-c}\, v\rangle_\cG + c\langle u,v\rangle_\cG, \quad u,v\in \qdom(T),
\end{gather*}
and the operator $T$ has  compact resolvent iff its form domain $\cQ(T)$ endowed with the above scalar product $\langle\cdot,\cdot\rangle_t\equiv \langle\cdot,\cdot\rangle_Q$
is compactly embedded into $\cG$.
It follows from the preceding discussion that a lower semibounded self-adjoint operator $T$ is uniquely determined by the knowledge of its form domain $\qdom(T)$
and of the diagonal values $T[u,u]$ of its sesquilinear form for all $u\in \qdom(T)$. Many operators appearing in the subsequent discussion
will be introduced in this way.

Using the above convention let us recall the min-max characterization of eigenvalues.
Let $T$ be a lower semibounded self-adjoint operator in an infinite-dimensional Hilbert space $\cG$.
For $j\in\NN$ we denote
\[
E_j(T):=\inf_{\substack{S\subset \qdom(T)\\ \dim S=j}} \sup_{\substack{u\in S\\u\ne 0}} \dfrac{T[u,u]}{\|u\|^2_\cG}.
\]
It follows from the min-max principle that $E_j(T)$ is the $j$th eigenvalue of $T$,
when enumerated in the non-decreasing order and counted with multiplicities, provided that it is
strictly below the bottom of the essential spectrum of $T$, and $E_1(T)$ coincides
with the bottom of the spectrum of $T$, see e.g. \cite[Section XIII.1]{RS}.
In particular, if $T$ has compact resolvent, then $E_j(T)$ is the $j$th eigenvalue of $T$ for any $j\in\NN$.
The main consequence of the min-max principle  we are going to use is as follows (the proof directly
follows from the definition):

\begin{prop}\label{prop-incl}
Let $T$ and $T'$ be lower semibounded self-adjoint operators in infinite-dimensional Hilbert spaces
$\cG$ and $\cG'$ respectively. Assume that there exists a linear map $J:\qdom(T)\to\qdom (T')$
such that $\|J u\|_{\cG'}=\|u\|_\cG$ and $T'[Ju,Ju]\le T[u,u]$ for all $u\in\qdom(T)$, then 
$E_j(T')\le E_j(T)$ for any $j\in\NN$.
\end{prop}

We will also use some classical results on the monotone convergence of operators. The following particular case which will be sufficient for our purposes:

\begin{prop}\label{prop-mon}
Let $\cH$ be a Hilbert space and $\cH_\infty$ be a closed subspace of $\cH$ endowed
with the induced scalar product. Let
\begin{itemize}
\item $T_n$ with $n\in\NN$ be lower semibounded
self-adjoint operators with compact resolvents in $\cH$,
\item $T_\infty$ be a lower semibounded
self-adjoint operator with compact resolvent in~$\cH_\infty$
\end{itemize}
such that the following conditions are satisfied:
\begin{itemize}
\item the sequence $(T_n)$  is monotonically increasing, i.e.
\[
\qdom(T_n)\supset \qdom(T_{n+1}), \quad T_n[u,u]\le T_{n+1}[u,u] \quad \text{for all $n\in\NN$ and $u\in \qdom (T_{n+1})$},
\]
\item one has the equalities
\begin{gather*}
\qdom(T_\infty)=\big\{u\in \bigcap\limits_{n\in \NN} \qdom(T_n): \quad \sup T_n[u,u]<\infty\big\},\\
T_\infty[u,u]=\lim_{n\to+\infty}T_n[u,u] \text{ for each } u\in \qdom(T_\infty),
\end{gather*}
\end{itemize}
then for each $j\in\NN$ there holds $E_j(T_\infty)=\lim_{n\to+\infty} E_j(T_n)$.
\end{prop}
The result follows, for example, from the constructions of \cite[Abs.~3]{weid}: Satz 3.1 establishes
a (generalized) strong resolvent convergence of $T_n$ to $T_\infty$ and Satz~3.2 gives the convergence of the eigenvalues.
An interested reader may refer to the papers \cite{bh,simon,weid}
dealing with the monotone convergence in a more general framework, i.e. beyond densely defined operators with compact resolvents.

\section{Sesquilinear forms}\label{sec-quad}

\subsection{Sesquilinear forms for the squares of Euclidean Dirac operators}\label{quad1}

For the rest of the text we denote
\[
\Omega^c:=\RR^n\setminus\overline{\Omega}.
\]
The shape operator $W:T\Sigma\to T\Sigma$ is given by $W X:=\nabla_X \nu$ with $\nabla$ being the gradient in $\RR^n$,
and its eigenvalues $h_1,\dots,h_{n-1}$ are the principal curvatures of $\Sigma$.
For $k=1,\dots,n-1$ we will denote by $H_k$ the \emph{$k$-th mean curvature of $\Sigma$ with respect to $\nu$}
defined by
\[
H_k=\sum_{1\le j_1<\dots<j_k\le n-1} h_{j_1}\cdot \ldots \cdot h_{j_k},
\]
in particular, $H_1=h_1+\ldots +h_{n-1}=\tr W$ is the mean curvature, $R=2H_2\equiv H_1^2-|W|^2$ with
$|W|^2:=\tr (W^2)$ is the scalar curvature. We set formally $H_k=0$ for $k\ge n$.

\begin{lemma}\label{qfa}
The operator $A_m$ is self-adjoint with  compact resolvent and its eigenfunctions belong to $C^\infty(\overline \Omega,\CC^N)$.
For all $u\in \dom(A_m)$
there holds
\begin{equation}
  \label{qform}
\langle A_m u, A_m u\rangle_{L^2(\Omega,\CC^N)}
=\int_{\Omega} \big(|\nabla u|^2 +m^2|u|^2\big)\dd x +\int_\Sigma\Big(m+ \dfrac{H_1}{2}\Big)\, |u|^2\dd s.
\end{equation}
\end{lemma}

\begin{proof}
Remark first that the map $x\mapsto \Gamma (x)$ in~\eqref{ekg1} gives a representation of the Clifford algebra~$\Cl(0,n)$.
Furthermore, the self-adjointness is not influenced if one adds a bounded operator, hence, 
it is sufficient to consider the case $m=0$. The operator $A_0$ 
is covered e.g. by the analysis of \cite[Section 2]{hmz} by noting that $\cB$ is a chirality operator defining a local boundary condition. Hence, the self-adjointness, the compactness of the resolvent and the smoothness of eigenfunctions
follow from \cite[Proposition~1 and Corollary~2]{hmz}. An interested reader may refer e.g. to \cite{baer}
for a more detailed discussion of boundary value problems for Dirac-type operators.

In order to obtain the representation  \eqref{qform} we use additional  constructions.
The map $\Gamma$ induces the extrinsic Dirac operator $\widetilde D^\Sigma$ in $L^2(\Sigma,\CC^N)$ given by
\[
\widetilde D^\Sigma \psi:=\dfrac{H_1}{2}\,\psi-\Gamma(\nu) \sum_{j=1}^{n-1} \Gamma(e_j) \nabla_{e_j}\psi
\]
with $(e_1,\dots,e_{n-1})$ being an orthonormal frame tangent to $\Sigma$. For  $u\in H^2(\Omega,\CC^N)$ one has the integral identity, see \cite[Section 3, Eq. (13)]{hmw},
\begin{equation*}
\int_\Omega |D_0 u|^2\dd x= \int_\Omega |\nabla u|^2\dd x
+\int_\Sigma \Big( \dfrac{H_1}{2}\, |u|^2 -\langle \widetilde D^\Sigma u, u\rangle\Big)\dd s,
\end{equation*}
where $D_0$ is given by  \eqref{eqedm} with $m=0$. Therefore, for $u\in H^2(\Omega,\CC^N)\cap\dom(A_m)$
one has
\begin{multline}
  \label{form1}
\langle A_m u,A_m u\rangle_{L^2(\Omega,\CC^N)}\equiv
\Big\langle \big(D_0 +m \alpha_{n+1}\big)u,\big(D_0 +m \alpha_{n+1}\big)u\Big\rangle_{L^2(\Omega,\CC^N)}\\
= \langle D_0 u,D_0 u\rangle_{L^2(\Omega,\CC^N)}
+2m\Re \Big(\big\langle D_0 u,\alpha_{n+1}u\big\rangle_{L^2(\Omega,\CC^N)}\Big)+m^2 \big\langle \alpha_{n+1}u,\alpha_{n+1}u\big\rangle_{L^2(\Omega,\CC^N)}\\
=\int_\Omega \big( |\nabla u|^2+m^2|u|^2\big)\dd x
+\int_\Sigma \Big( \dfrac{H_1}{2}\, |u|^2 -\langle \widetilde D_\Sigma u, u\rangle\Big)\dd s\\
+2m\Re\Big( \langle D_0 u,\alpha_{n+1}u\rangle_{L^2(\Omega,\CC^N)}\Big).
\end{multline}
The operator $\widetilde D^\Sigma$ anticommutes with $\Gamma(\nu)$, see \cite[Proposition 1]{hmw}.
As the matrix $\alpha_{n+1}$
anticommutes with all $\Gamma(x)$, it commutes with  $\widetilde D^\Sigma$ by construction.
Therefore, using the boundary condition for $u$ we have the pointwise equalities
\begin{align*}
\langle \widetilde D^\Sigma u, u\rangle&=\big\langle \widetilde D^\Sigma \big[-\rmi \alpha_{n+1}\Gamma(\nu) \big] u, u\big\rangle\\
&=\big\langle \rmi \alpha_{n+1}\Gamma(\nu)  \widetilde D^\Sigma u, u
\big\rangle
=\big\langle   \widetilde D^\Sigma u, -\rmi\Gamma(\nu)\alpha_{n+1} u\big\rangle\\
&=\big\langle  \widetilde D^\Sigma u, \rmi\alpha_{n+1}\Gamma(\nu) u\big\rangle
=-\langle \widetilde D^\Sigma u, u\rangle,
\end{align*}
implying $\langle \widetilde D^\Sigma u, u\rangle=0$ on $\Sigma$.

It remains to transform the third summand on the right-hand side of \eqref{form1}.
Recall that due to the integration by parts for any $v,w\in H^1(\Omega,\CC^N)$ we have
\[
\int_\Omega \sum_{j=1}^n \langle \alpha_j \partial_j v,w\rangle_{\CC^N}\dd x
=-\int_\Omega \sum_{j=1}^n \langle  v, \alpha_j \partial_j w\rangle_{\CC^N}\dd x
+\int_\Sigma \sum_{j=1}^n \langle  \alpha_j \nu_j v, w\rangle_{\CC^N}\dd s,
\]
which then gives
\begin{multline}
   \label{dpart}
\big\langle D_0 u,\alpha_{n+1}u\big\rangle_{L^2(\Omega,\CC^N)}=
\int_\Omega \big\langle D_0 u,\alpha_{n+1}u\big\rangle_{\CC^N}\dd x\\
=\int_\Omega \big\langle  u, D_0\alpha_{n+1}u\big\rangle_{\CC^N}\dd x
+\int_\Sigma \sum_{j=1}^n \langle  -\rmi\alpha_j \nu_j u, \alpha_{n+1} u\rangle_{\CC^N}\dd s\\
=-\int_\Omega \big\langle  \alpha_{n+1}u, D_0 u\big\rangle_{\CC^N}\dd x
+\int_\Sigma \big\langle -\rmi\Gamma(\nu)u,\alpha_{n+1}u\big\rangle_{\CC^N}\dd s.
\end{multline}
Therefore, 
\begin{align*}
2m\Re\Big(\big\langle D_0 u,\alpha_{n+1}u\big\rangle_{L^2(\Omega,\CC^N)}\Big)&=
m\Big(\big\langle D_0 u,\alpha_{n+1}u\big\rangle_{L^2(\Omega,\CC^N)}
+ \big\langle  \alpha_{n+1}u, D_0 u\big\rangle_{L^2(\Omega,\CC^N)}
\Big)\\
&=m\int_\Sigma \big\langle -\rmi\Gamma(\nu)u,\alpha_{n+1}u\big\rangle_{\CC^N}\dd s\\
&=m \int_\Sigma \big\langle -\rmi\alpha_{n+1} \Gamma(\nu)u,u\big\rangle_{\CC^N}\dd s
=m\int_\Sigma |u|^2_{\CC^N}\dd s.
\end{align*}
This shows the sought identity \eqref{qform} for the $H^2$ functions in the domain.
It is then extended to the whole of $\dom(A_m)$ by a standard density argument.
\end{proof}

\begin{lemma}
The operator $B_{m,M}$ is self-adjoint, and for all $u\in \dom(B_{m,M})$ there holds
\begin{multline}
  \label{qform2}
\langle B_{m,M} u, B_{m,M} u\rangle_{L^2(\RR^n,\CC^N)}
=\int_{\Omega} \big(|\nabla u|^2 +m^2|u|^2\big)\dd x
+ \int_{\Omega^c} \big(|\nabla u|^2 +M^2|u|^2\big)\dd x\\
+(M-m)\Big(\int_\Sigma | \cP_- u|^2\dd s -\int_\Sigma | \cP_+ u|^2\dd s\Big),
\end{multline}
where $\cP_\pm(s):=\dfrac{I_N \pm \cB(s)}{2}$ for $s\in\Sigma$.
\end{lemma}

\begin{proof}
The self-adjointness is obvious with the help of the Fourier transform, so let us concentrate on the sesquilinear form.
Representing $B_{m,M}=D_M +(m-M) 1_\Omega \alpha_{n+1}$ we have
\begin{multline*}
\langle B_{m,M} u, B_{m,M} u\rangle_{L^2(\RR^n,\CC^N)}\\
\begin{aligned}
&=\langle D_M u +(m-M) 1_\Omega \alpha_{n+1}u, D_M u +(m-M) 1_\Omega \alpha_{n+1} u\rangle_{L^2(\RR^n,\CC^N)}\\
&=\langle D_M u, D_M u\rangle_{L^2(\RR^n,\CC^N)}+(m-M)^2 \langle 1_\Omega \alpha_{n+1}u, 1_\Omega \alpha_{n+1} u\rangle_{L^2(\RR^n,\CC^N)}\\
&\quad+(m-M)\Big(
\langle D_M u, 1_\Omega \alpha_{n+1} u\rangle_{L^2(\RR^n,\CC^N)}+\langle 1_\Omega \alpha_{n+1}u, D_M u\rangle_{L^2(\RR^n,\CC^N)}
\Big)\\
&=\int_{\RR^n} \big(|\nabla u|^2 +M^2|u|^2\big)\dd x + (m-M)^2\int_\Omega |u|^2\dd x\\
&\quad +
(m-M)\Big(
\langle D_M u, 1_\Omega \alpha_{n+1} u\rangle_{L^2(\RR^n,\CC^N)}+\langle 1_\Omega \alpha_{n+1}u, D_M u\rangle_{L^2(\RR^n,\CC^N)}
\Big).
\end{aligned}
\end{multline*}
Now using $D_M=D_0+M\alpha_{n+1}$ we transform the last summand as follows:
\begin{multline*}
(m-M)\Big[
\langle D_M u, 1_\Omega \alpha_{n+1} u\rangle_{L^2(\RR^n,\CC^N)}+\langle 1_\Omega \alpha_{n+1}u, D_M u\rangle_{L^2(\RR^n,\CC^N)}\Big]\\
\begin{aligned}
&= (m-M)\Big[
\langle D_0 u+M\alpha_{n+1} u, 1_\Omega \alpha_{n+1} u\rangle_{L^2(\RR^n,\CC^N)}\\
&\qquad +\langle 1_\Omega \alpha_{n+1}u,  D_0 u+M\alpha_{n+1} u\rangle_{L^2(\RR^n,\CC^N)}\Big]\\
&=2M(m-M)\int_\Omega |u|^2\dd x\\
&\qquad  + (m-M)\Big(
\langle D_0 u, 1_\Omega \alpha_{n+1} u\rangle_{L^2(\RR^n,\CC^N)}+\langle 1_\Omega \alpha_{n+1}u,  D_0 u\rangle_{L^2(\RR^n,\CC^N)}\Big)\\
&=2M(m-M)\int_\Omega |u|^2\dd x+(m-M)\int_\Sigma \langle \cB u, u\rangle_{\CC^N}\dd s,
\end{aligned}
\end{multline*}
where  we used the equality \eqref{dpart} in the last step. This gives
\begin{multline*}
\langle B_{m,M} u, B_{m,M} u\rangle_{L^2(\RR^n,\CC^N)}
=\int_{\Omega} \big(|\nabla u|^2 +m^2|u|^2\big)\dd x\\
+ \int_{\Omega^c} \big(|\nabla u|^2 +M^2|u|^2\big)\dd x
-(M-m)\int_\Sigma \langle \cB u, u\rangle_{\CC^N}\dd s,
\end{multline*}
and it remains to remark that
\begin{multline*}
\langle \cB u, u\rangle_{\CC^N}=\dfrac{1}{2}\,\big\langle (1+\cB) u, u\big\rangle_{\CC^N}- \dfrac{1}{2}\,\big\langle (1-\cB) u, u\big\rangle_{\CC^N}\\
=\langle \cP_+ u, u\rangle_{\CC^N}-\langle \cP_- u, u\rangle_{\CC^N}\equiv |\cP_+u|_{\CC^N}-|\cP_-u|_{\CC^N},
\end{multline*}
where in the last step we used the fact that $\cP_\pm$ are orthogonal projectors.
\end{proof}

\subsection{Intrinsic and extrinsic Dirac operators on Euclidean hypersurfaces}\label{quad2}

The definition of the intrinsic Dirac operator $\Dsl$ on $\Sigma$ with a detailed presentation of preliminary constructions can be found in the monographs \cite{moroianu, fried,ginoux}.
Recall that if $\mathbb{S}\Sigma$ is the intrinsic spinor bundle over $\Sigma$ with the associated spin connection
$\dnabla$ and carrying the natural Hermitian and Clifford module structures, then $\Dsl$ acts on smooth sections $\psi$ of $\mathbb{S}\Sigma$
by $\Dsl \psi=\sum_{j=1}^{n-1} e_j\cdot \dnabla_{e_j} \psi$, where $(e_1,\dots,e_{n-1})$ is an orthonormal frame tangent to $\Sigma$
and $\cdot$ is the Clifford multiplication. For our situation, the study of $\Dsl$ is easier to approach through the so-called extrinsic Dirac operators, which will be more suitable for the subsequent asymptotic analysis, and we explain this link in the present section.

For $n\ge 2$ and $K:=2^{[\frac{n}{2}]}$ let $\beta_1,\dots,\beta_n$ be anticommuting Hermitian $K\times K$ matrices with $\beta_j^2=I_K$.
The intrinsic Dirac operator $D^{\RR^n}$ in $\RR^n$ acts then by
\[
D^{\RR^n}=-\rmi\sum_{j=1}^n \beta_j \dfrac{\partial}{\partial x_j},
\]
and it is a self-adjoint operator in $L^2(\RR^n,\CC^K)$ with  domain $H^1(\RR^n,\CC^K)$. Remark that the expression $D_0$ given in the introduction
does not correspond to the intrinsic Dirac operator on $\RR^n$ as $N\ne K$ in general.
The \emph{extrinsic} Dirac operator $D^\Sigma$ on $\Sigma$ is a self-adjoint operator in $L^2(\Sigma,\CC^K)$ with  domain $H^1(\Sigma,\CC^K)$ and given by
\[
D^\Sigma=\dfrac{H_1}{2}-\beta(\nu) \sum_{j=1}^{n-1} \beta(e_j) \nabla_{e_j},
\]
where $(e_1,\dots,e_{n-1})$ is an orthonormal frame tangent to $\Sigma$, and for $x=(x_1,\dots,x_n)$ we denote
$\beta(x)=\sum_{j=1}^n \beta_j x_j$. It is a fundamental result that
$D^\Sigma$ is unitarily equivalent to $\Dsl$ for odd $n$ and  to $\Dsl\oplus (-\Dsl)$ for even $n$;
for even $n$ the operator $\Dsl$ can be identified with the restriction of $\beta(\nu) D^\Sigma$ on $\ker\big(1-\beta(\nu)\big)$,
see e.g. \cite[Section 2.4]{moroianu}. In other words, the study of the eigenvalues of $(D^\Sigma)^2$ is equivalent to that of $\Dsl^2$, modulo the multiplicities for even $n$.

In turn, a classical tool for the analysis of the eigenvalues of $(D^\Sigma)^2$ is provided by the Schr\"odinger-Lichnerowicz formula
$(D^\Sigma)^2=(\nabla^\Sigma)^*\nabla^\Sigma+\frac{1}{2} \,H_2\, I$ (whose proof we recall in Appendix~\ref{sec-lichn}),
where $\nabla^\Sigma$ is the induced spin connection
\[
\nabla^\Sigma_X=\nabla_X+\dfrac{1}{2}\,\beta(\nu)\beta(WX):\, C^\infty(\Sigma,\CC^K)\to C^\infty(\Sigma,\CC^K), \quad
X\in T\Sigma.
\]
In other words, for $u\in H^1(\Sigma,\CC^K)$ one has
\begin{equation}
  \label{lichn}
\langle D^\Sigma u,D^\Sigma u\rangle_{L^2(\Sigma,\CC^K)}
=\int_\Sigma \Big(|\nabla^\Sigma u|^2 + \dfrac{H_2 |u|^2}{2} \Big)\dd x,
\end{equation}
while in the local coordinates on $\Sigma$ one has
\begin{equation}
  \label{lichn2}
|\nabla^\Sigma u|^2=\sum_{j,k=1}^{n-1} g^{jk} \Big\langle
\partial_j u +\dfrac{1}{2}\,\beta(\nu)\beta(\partial_j\nu) u,\partial_k u +\dfrac{1}{2}\,\beta(\nu)\beta(\partial_k\nu) u
\Big\rangle_{\CC^K},
\end{equation}
where $(g^{jk}):=(g_{jk})^{-1}$ and $(g_{jk})$ is the Riemannian metric on $\Sigma$ induced by the embedding into $\RR^n$.

For the subsequent analysis we introduce the Hilbert space
\begin{equation}
\cH:=\Big\{ f\in L^2(\Sigma,\CC^N):\, f=\cB f\Big\},
\quad
\|f\|^2_\cH:=\int_\Sigma |f|^2\dd s,
\end{equation}
with $\cB$ given in \eqref{ekg3}, and the self-adjoint operator $L$ in $\cH$ given by its sesquilinear form
as follows:
\[
L[f,f]=\int_\Sigma \Big[|\nabla f|^2 +\Big(H_2-\dfrac{H_1^2}{4}\Big) |f|^2\Big]\dd s,
\quad
\qdom(L)=H^1(\Sigma,\CC^N)\cap \cH,
\]
with $\qdom(L)$ being the form domain (see Section~\ref{sec-prelim}).
The operator $L$ will arise naturally in the asymptotic spectral analysis of the Dirac operators $A_m$ and $B_{m,M}$,
and its importance is explained in the following assertion:

\begin{lemma}\label{lemld}
The operator $L$ is unitarily equivalent to $\Dsl^2$.
\end{lemma}

\begin{proof}
The proof is by direct computation, by constructing an explicit isomorphism between
$L^2(\Sigma,\CC^{N/2})$ and $\cH$ and then by establishing a link with the extrinsic Dirac operator $D^\Sigma$
using the Schr\"odinger-Lichnerowicz formula.

Following the standard rules, see e.g. \cite[Chapter 15]{dg} or \cite[Appendix~E]{wit}, for $n\in\NN$
we define $2^{[\frac{n}{2}]}\times 2^{[\frac{n}{2}]}$ Dirac matrices $\gamma_j(n)$ with $j\in\{1,\dots,n\}$
using the following  iterative procedure:
\begin{itemize}
\item For $n=1$, set $\gamma_1(1):=(1)$.
\item For $n=2$, set $\gamma_1(2):=\begin{pmatrix}
0 & 1 \\ 1 & 0
\end{pmatrix}$ and $\gamma_2(2):=\begin{pmatrix}
0 & -\rmi \\ \rmi & 0
\end{pmatrix}$.
\item For $n=2m+1$ with $m\in\NN$:
\begin{align}
\gamma_j(2m+1)&:=\gamma_j(2m), \quad j=1,\dots,2m, \nonumber\\
\gamma_{2m+1}(2m+1)&:=\pm\rmi^m \gamma_1(2m)\cdot \ldots \cdot \gamma_{2m}(2m)=\pm\begin{pmatrix} -I_{2^{m-1}} & 0 \\ 0 & I_{2^{m-1}} \end{pmatrix},  \label{eqpmg}
\end{align}
\item For $n=2m+2$ with $m\in\NN$:
\begin{align*}
\gamma_j(2m+2)&:=\begin{pmatrix}
0 & \gamma_j(2m+1)\\
\gamma_j(2m+1) & 0
\end{pmatrix}, \quad j=1,\dots,2m+1,\nonumber\\
\gamma_{2m+2}(2m+2)&:=\begin{pmatrix}
0 & -\rmi I_{2^{m}}\\
\rmi  I_{2^{m}}& 0
\end{pmatrix} . \nonumber
\end{align*}
\end{itemize}
One easily checks that at a fixed $n\in\NN$ the matrices $\gamma_j(n)$ are Hermitian and anticommute, the square of each of them is the identity matrix.
Furthermore, if $\big(\gamma'_j(n)\big)$ is another set of matrices with these properties and of the same size, then there exists a unitary matrix $C$
and a suitable choice of $\pm$ in \eqref{eqpmg} such that the equalities $\gamma'_j(n) C=\gamma_j(n) C$ hold for all $j$, see e.g. \cite[Prop.~15.16]{dg}. Therefore, without loss of generality one may assume that
the matrices $\alpha_j$ in the expression \eqref{ekg3} of $\cB$ and
the matrices $\beta_j$ used in the definition of $D^\Sigma$ are chosen in the form
\begin{equation}
  \label{agan}
\alpha_j=\gamma_j(n+1), \quad j=1,\dots,n+1,
\qquad
\beta_j=\gamma_j(n), \quad j=1,\dots,n.
\end{equation}
For $x=(x_1,\dots,x_n)\in\RR^n$ and $q\in\{n,n+1\}$ we define a matrix $\Gamma_q(x)$ by
\[
\Gamma_q(x)=\sum_{j=1}^n x_j \gamma_j(q),
\]
then one has the relations
\begin{gather}
   \label{eq-comm-rel}
\Gamma_n(x)\Gamma_n(y)+\Gamma_n(y)\Gamma_n(x)=2\langle x,y\rangle_{\RR^n}I, \quad x,y\in\RR^n,\\
\Gamma(x)=\Gamma_{n+1}(x), \quad \beta(x)=\Gamma_n(x).\nonumber
\end{gather}

Consider first the case when $n$ is odd, $n=2m+1$ with $m\in\NN$.
Represent $f\in\cH$ as $f=(f_-,f_+)$ with $f_\pm\in L^2(\Sigma,\CC^{N/2})$,
then, under the convention \eqref{agan}, the condition $f=\cB f$
takes the form
\begin{gather*}
\begin{pmatrix}
f_- \\ f_+
\end{pmatrix}=-\rmi \begin{pmatrix} 0 & -\rmi I_{2^m}\\ \rmi I_{2^m}& 0\end{pmatrix}
\begin{pmatrix}
0 & \Gamma_n(\nu) \\
\Gamma_n(\nu) & 0
\end{pmatrix}
\begin{pmatrix}
f_- \\ f_+
\end{pmatrix}, 
\end{gather*}
which holds if and only if $f_\pm=\pm \Gamma_n(\nu)f_\pm$. Therefore, the map
\[
U:L^2(\Sigma,\CC^{N/2})\to \cH, \quad (U f)(s)=\dfrac{1}{2}\begin{pmatrix} \big(1-\Gamma_n(\nu) \big) f\\
\big(1+\Gamma_n(\nu) \big) f\end{pmatrix}
\]
defines a unitary operator, and $Uf \in H^1(\Sigma,\CC^N)$ iff $f\in H^1(\Sigma,\CC^{N/2})$.
As $H_j$ are scalar functions, one has
\begin{equation}
 \label{h1h2}
\Big(H_2-\dfrac{H_1^2}{4}\Big) |U f|^2_{\CC^N} = 
\Big(H_2-\dfrac{H_1^2}{4}\Big) |f|^2_{\CC^{N/2}}.
\end{equation}
In order to compute $\big|\nabla (Uf)\big|^2$ we use local coordinates on $\Sigma$.
One has
\begin{align*}
\big|\nabla (Uf)\big|^2&= \dfrac{1}{4}\sum_{j,k=1}^{n-1}
g^{j,k}\bigg[
\Big\langle \partial_j \Big(\big(1-\Gamma_n(\nu) \big) f\Big),\partial_k \Big(\big(1-\Gamma_n(\nu) \big) f\Big)\Big\rangle_{\CC^{N/2}}
\\
&\qquad+ \Big\langle \partial_j \Big(\big(1+\Gamma_n(\nu) \big) f\Big),\partial_k \Big(\big(1+\Gamma_n(\nu) \big) f\Big)\Big\rangle_{\CC^{N/2}}
\bigg]\\
&= \dfrac{1}{2}\sum_{j,k=1}^{n-1}g^{j,k}\bigg[
\langle \partial_j f,\partial_k f\rangle_{\CC^{N/2}}
+
\Big\langle
\partial_j \big(\Gamma_n(\nu) f\big)
,
\partial_k \big(\Gamma_n(\nu) f\big)
\Big\rangle_{\CC^{N/2}}\bigg].
\end{align*}
We have then
\begin{multline*}
\Big\langle
\partial_j \big(\Gamma_n(\nu) f\big),
\partial_k \big(\Gamma_n(\nu) f\big)
\Big\rangle_{\CC^{N/2}}\\
=\Big\langle \Gamma_n(\nu)\partial_j f + \Gamma_n(\partial_j\nu) f,\Gamma_n(\nu)\partial_k f + \Gamma_n(\partial_k\nu) f\Big\rangle_{\CC^{N/2}}\\
=\Big\langle \partial_j f + \Gamma_n(\nu)\Gamma_n(\partial_j\nu) f,\partial_k f + \Gamma_n(\nu)\Gamma_n(\partial_k\nu) f\Big\rangle_{\CC^{N/2}},
\end{multline*}
and it follows that
\begin{align*}
\big|\nabla (Uf)\big|^2&=\sum_{j,k=1}^{n-1} g^{j,k} \Big\langle\partial_j f + \dfrac{1}{2}\,\Gamma_n(\nu) \Gamma_n(\partial_j\nu) f, \partial_k f + \dfrac{1}{2}\,\Gamma_n(\nu) \Gamma_n(\partial_k\nu) f\Big\rangle_{\CC^{N/2}}\\
&\quad +\frac14\sum_{j,k=1}^{n-1} g^{j,k} \Big\langle\Gamma_n(\partial_k \nu)\Gamma_n(\partial_j\nu)f,f\Big\rangle_{\CC^{N/2}}\\
&= |\nabla^\Sigma f|^2+\dfrac{1}{4}\langle f, V f\rangle_{\CC^{N/2}}, \qquad V:=\sum_{j,k=1}^{n-1} g^{j,k} \Gamma_n(\partial_k \nu)\Gamma_n(\partial_j\nu).
\end{align*}
Using the symmetry of $(g^{j,k})$ and the commutation relation \eqref{eq-comm-rel} we compute
\begin{multline*}
V=\dfrac{1}{2}\sum_{j,k=1}^{n-1} g^{j,k} \Big(\Gamma_n(\partial_j \nu)\Gamma_n(\partial_k\nu)+\Gamma_n(\partial_k \nu)\Gamma_n(\partial_j\nu)\Big)\\
= \sum_{j,k=1}^{n-1} g^{j,k} \langle\partial_j \nu,\partial_k\nu\rangle\, I= |\nabla\nu|^2 I=|W|^2 I=(H_1^2-2H_2) I.
\end{multline*}
By combining with \eqref{h1h2} we arrive at
\begin{gather*}
L[Uf,Uf]=\int_\Sigma \Big( |\nabla^\Sigma f|^2+\dfrac{H_2 |f|^2}{2} \Big)\dd s.
\end{gather*}
Due to the Schr\"odinger-Lichnerowicz formula \eqref{lichn} we conclude that $L=U^*(D^\Sigma)^2 U$,
while $(D^\Sigma)^2$ is unitarily equivalent to $\Dsl^2$ as $n$ is odd. This proves the claim for odd dimensions.

Now consider the case when $n$ is even, $n=2m$ with $m\in\NN$. As for the previous case, we try to find a block representation
for the  condition $f=\cB f$, which now takes the form
\begin{equation}
   \label{loc0}
\Big(I_{2^m} +\rmi \gamma_{2m+1}(2m+1) \sum_{j=1}^{2m} \gamma_j(2m+1)\, \nu_j\Big)f=0.
\end{equation}
We first remark that for $x=(x_1,\dots,x_n)\in\RR^n$ we have the block representation
\begin{gather}
\sum_{j=1}^{2m} \gamma_j(2m+1)x_j\equiv \sum_{j=1}^{2m} \gamma_j(2m)x_j\equiv\Gamma_n (\nu)= \begin{pmatrix} 0 & \lambda(x) \\
\lambda(x)^* & 0 \end{pmatrix},  \label{lambdas}\\
\lambda(x):=\sum_{j=1}^{2m-1} \gamma_j(2m-1) \, x_j -\rmi x_{2m} I_{2^{m-1}}.
\nonumber
\end{gather}
Represent $f=(\psi_-,\psi_+)$ with $\psi_\pm\in L^2(\Sigma,\CC^{N/2})$, then we rewrite
the condition \eqref{loc0} in the block form
\[
\left[
\begin{pmatrix} I & 0 \\ 0 & I\end{pmatrix}
\pm\rmi \begin{pmatrix} -I & 0 \\ 0 & I \end{pmatrix} \begin{pmatrix} 0 & \lambda(\nu) \\ \lambda(\nu)^* & 0\end{pmatrix}
\right]\begin{pmatrix} \psi_- \\ \psi_+\end{pmatrix}=\begin{pmatrix} 0 \\ 0 \end{pmatrix},
\]
where $I:=I_{2^{m-1}}$.
Using $\lambda(\nu)\lambda(\nu)^*=\lambda(\nu)^*\lambda(\nu)=I$ we see that
the condition $f=\cB f$ can be rewritten as $\psi_-=\pm\rmi \lambda(\nu)\,\psi_+$.
Hence, the map
\[
U: L^2(\Sigma,\CC^{N/2})\to \cH,
\quad
U \psi
=
\dfrac{1}{\sqrt{2}}\begin{pmatrix}
\pm\rmi \lambda(\nu) \psi \\ \psi
\end{pmatrix}
\]
defines a unitary operator, and at each point of $\Sigma$ there holds
\begin{equation}
  \label{equu1}
\begin{aligned}
\big|\nabla (U\psi)\big|^2&=
\sum_{j,k=1}^{n-1} g^{j,k} \bigg(\dfrac{1}{2}\,
\Big\langle \rmi\lambda(\nu)\partial_j \psi + \rmi\lambda(\partial_j\nu) \psi,\rmi\lambda(\nu)\partial_k \psi + \rmi\lambda(\partial_k\nu) \psi\Big\rangle_{\CC^{N/2}}\\
&\quad+\dfrac{1}{2}\,\langle\partial_j \psi,\partial_k\psi\rangle_{\CC^{N/2}}\bigg).
\end{aligned}
\end{equation}
We then transform
\begin{multline*}
\dfrac{1}{2}\,\Big\langle \rmi\lambda(\nu)\partial_j \psi + \rmi\lambda(\partial_j\nu) \psi,\rmi\lambda(\nu)\partial_k \psi + \rmi\lambda(\partial_k\nu) \psi\Big\rangle_{\CC^{N/2}}+\dfrac{1}{2}\langle\partial_j \psi,\partial_k\psi\rangle_{\CC^{N/2}}\\
\begin{aligned}
&=\dfrac{1}{2}\,\Big\langle \partial_j \psi + \lambda(\nu)^*\lambda(\partial_j\nu) \psi,\partial_k \psi + \lambda(\nu)^*\lambda(\partial_k\nu) \psi\Big\rangle_{\CC^{N/2}}+\dfrac{1}{2}\langle\partial_j \psi,\partial_k\psi\rangle_{\CC^{N/2}}\\
&=\Big\langle \partial_j \psi + \dfrac{1}{2}\lambda(\nu)^*\lambda(\partial_j\nu) \psi,\partial_k \psi + \dfrac{1}{2}\lambda(\nu)^*\lambda(\partial_k\nu) \psi\Big\rangle_{\CC^{N/2}}\\
&\quad+\dfrac{1}{4}\,\Big\langle \lambda(\nu)^*\lambda(\partial_j\nu) \psi,\lambda(\nu)^*\lambda(\partial_k\nu) \psi\Big\rangle_{\CC^{N/2}}\\
&=\Big\langle \partial_j \psi + \dfrac{1}{2}\lambda(\nu)^*\lambda(\partial_j\nu) \psi,\partial_k \psi + \dfrac{1}{2}\lambda(\nu)^*\lambda(\partial_k\nu) \psi\Big\rangle_{\CC^{N/2}}\\
&\quad+\dfrac{1}{4}\,\Big\langle  \psi,\lambda(\partial_j\nu)^*\lambda(\partial_k\nu) \psi\Big\rangle_{\CC^{N/2}}.
\end{aligned}
\end{multline*}
The substitution into \eqref{equu1} gives
\begin{align*}
\big|\nabla (U\psi)\big|^2&=\sum_{j,k=1}^{n-1} g^{j,k}\Big\langle \partial_j \psi + \dfrac{1}{2}\lambda(\nu)^*\lambda(\partial_j\nu) \psi,\partial_k \psi + \dfrac{1}{2}\lambda(\nu)^*\lambda(\partial_k\nu) \psi\Big\rangle_{\CC^{N/2}}\\
&\quad+\dfrac{1}{4}\langle \psi, V \psi\rangle_{\CC^{N/2}},
\qquad V:=\sum_{j,k=1}^{2m-1}g^{j,k}\lambda(\partial_j\nu)^*\lambda(\partial_k\nu).
\end{align*}
In order to compute $V$ we introduce
\[
\widetilde V:=\sum_{j,k=1}^{2m-1}g^{j,k}\lambda(\partial_j\nu)\lambda(\partial_k\nu)^*,
\]
then
\begin{align*}
\begin{pmatrix}
\widetilde V & 0\\
0 &  V
\end{pmatrix}
&=\sum_{j,k=1}^{2m-1}g^{j,k}\begin{pmatrix}\lambda(\partial_j\nu)\lambda(\partial_k\nu)^* & 0 \\ 0 &\lambda(\partial_j\nu)^*\lambda(\partial_k\nu)\end{pmatrix}\\
&=\sum_{j,k=1}^{2m-1}g^{j,k}\begin{pmatrix}0 & \lambda(\partial_j\nu) \\ \lambda(\partial_j\nu)^* &0\end{pmatrix}
\begin{pmatrix}0 & \lambda(\partial_k\nu) \\ \lambda(\partial_k\nu)^* & 0\end{pmatrix}\\
&= \sum_{j,k=1}^{2m-1}g^{j,k} \Gamma_n(\partial_j \nu)\Gamma_n(\partial_k\nu)\\
&=\dfrac{1}{2}
\sum_{j,k=1}^{2m-1}g^{j,k} \Big( \Gamma_n(\partial_j \nu)\Gamma_n(\partial_k\nu)+\Gamma_n(\partial_k \nu)\Gamma_n(\partial_j\nu)\Big)\\
&=\sum_{j,k=1}^{2m-1}g^{j,k} \langle \partial_j \nu,\partial_k\nu\rangle\, I= |\nabla \nu|^2 I=|W|^2I=(H_1^2-2H_2)I.
\end{align*}
In addition, as the functions $H_j$ are scalar, we have
\[
\Big\langle U\psi, \Big(H_2-\dfrac{H_1^2}{4}\Big)U\psi\Big\rangle_{\cH}
=\Big\langle \psi, \Big(H_2-\dfrac{H_1^2}{4}\Big)\psi\Big\rangle_{L^2(\Sigma,\CC^{N/2})},
\]
and then
\begin{multline*}
L[U\psi,U\psi]\\
=\int_\Sigma \sum_{j,k=1}^{n-1} g^{j,k}\Big\langle \partial_j \psi + \dfrac{1}{2}\lambda(\nu)^*\lambda(\partial_j\nu) \psi,\partial_k \psi + \dfrac{1}{2}\lambda(\nu)^*\lambda(\partial_k\nu) \psi\Big\rangle_{\CC^{N/2}}\dd s\\
+\dfrac{1}{2}\langle \psi, H_2 \psi\rangle_{L^2(\Sigma,\CC^{N/2})}.
\end{multline*}
Now consider the unitary transform $U_0:L^2(\Sigma,\CC^{N/2})\to L^2(\Sigma,\CC^{N/2})$ given by $U_0\psi=\lambda(\nu)^* \psi$,
then a simple computation shows that
\begin{multline*}
L[UU_0\psi,UU_0\psi]\\
=\int_\Sigma \sum_{j,k=1}^{n-1} g^{j,k}\Big\langle \partial_j \psi + \dfrac{1}{2}\lambda(\nu)\lambda(\partial_j\nu)^* \psi,\partial_k \psi + \dfrac{1}{2}\lambda(\nu)\lambda(\partial_k\nu)^* \psi\Big\rangle_{\CC^{N/2}}\dd s\\
+\dfrac{1}{2}\langle \psi, H_2 \psi\rangle_{L^2(\Sigma,\CC^{N/2})}.
\end{multline*}
Using \eqref{lambdas}, for $\psi_\pm\in H^1(\Sigma,\CC^{N/2})$ and $\psi:=(\psi_-,\psi_+)\in H^1(\Sigma,\CC^N)$ one has
\begin{multline*}
L[UU_0\psi_-,UU_0\psi_-]+L[U\psi_+,U\psi_+]\\
=\int_\Sigma \sum_{j,k=1}^{n-1} g^{j,k}
\big\langle\partial_j \psi + \dfrac{1}{2} \Gamma_n(\nu)\Gamma_n(\partial_j\nu)\psi,
\partial_k \psi + \dfrac{1}{2} \Gamma_n(\nu)\Gamma_n(\partial_k\nu)\psi\big\rangle_{\CC^N}\dd s\\
+\dfrac{1}{2}\, \langle \psi, H_2 \psi\rangle_{L^2(\Sigma,\CC^N)}. 
\end{multline*}
By comparing with the Schr\"odinger-Lichnerowicz formula \eqref{lichn}--\eqref{lichn2}
we see that the operator $(U_0^*U^*LUU_0)\oplus (U^*LU)$
is unitarily equivalent to $(D^\Sigma)^2$. As $(D^\Sigma)^2$ is now unitarily equivalent
to $\Dsl^2\oplus\Dsl^2$ (because $n$ is even), it follows that
$L$ is unitarily equivalent to $\Dsl^2$.
\end{proof}

\section{Preliminary constructions for the spectral analysis}\label{sec-prelim}

\subsection{One-dimensional model operators}\label{ssec1d}

\begin{lemma}\label{lem1dd}
Let $\delta>0$ be fixed. For $\alpha>0$, let $S$ be the self-adjoint operator
in $L^2(0,\delta)$ with 
\[
S[f,f]=\int_0^\delta |f'|^2\dd t-\alpha \big|f(0)\big|^2,
\quad
\qdom(S)=\big\{f\in H^1(0,\delta):\, f(\delta)=0\big\},
\]
then for $\alpha\to+\infty$ one has $E_1(S)=-\alpha^2+\cO(e^{-\delta\alpha})$,
and the associated eigenfunction $\psi$ with $\|\psi\|_{L^2(0,\delta)}=1$
satisfies $\big|\psi(0)\big|^2=2\alpha+\cO(1)$.
\end{lemma}
\begin{proof}
One easily see that the operator $S$ acts as $f\to-f''$ defined of the functions $f\in H^2(0,\delta)$ with $f'(0)+\alpha f(0)=f(\delta)=0$.
Let us estimate its first eigenvalue as $\alpha\to+\infty$. Look for negative eigenvalues $E=-k^2$ with $k>0$,
then using the boundary condition at $\delta$ we see that the associated normalized
eigenfunction $\psi$ is of the form $\psi(t)=c\sinh\big(k(\delta-t)\big)$ with $c\ne 0$ being a normalizing constant.
The boundary condition at $0$
gives $0=\psi'(0)+\alpha \psi(0)=-k \cosh (k\delta)+\alpha \sinh (k\delta)$, i.e.
\begin{equation}
 \label{fkd}
F(k\delta)=\alpha\delta, \quad F(x):= x\coth x.
\end{equation}
One easily sees that $F:(0,+\infty)\to (1,+\infty)$ is strictly increasing and bijective,
and for $\alpha\delta>1$ the equation \eqref{fkd} admits a unique solution $k$, and then $k\delta\to +\infty$ for $\alpha\to+\infty$. 
Now rewrite \eqref{fkd} as $k=\alpha \tanh (k\delta)$. Due to $k\delta\to+\infty$ we have $\frac{3}{4}\le \tanh (k\delta)\le 1$
implying $3\alpha/4\le k\le \alpha$. Then using the equation again we have $\alpha \tanh\big(\frac{3}{4} \,\alpha\delta\big)\le k\le \alpha$,
while $\tanh\big(\frac{3}{4} \,\alpha \delta\big)=1+\cO(e^{-3 \delta \alpha/2})$. Therefore, with some $c_1>0$
one has $E_1(S)=-k^2=-\alpha^2\big(1+\cO(e^{-3\delta\alpha/2})\big)\le -\alpha^2+c_1e^{-\delta\alpha}$ as $\alpha\to +\infty$.

In order to find the value of the normalizing constant $c$ we use
\[
1=\|\psi\|^2_{L^2(0,\delta)}=|c|^2\int_0^\delta \sinh^2\big( k(\delta-t)\big)\dd t=|c|^2 \Big(\dfrac{1}{4 k} \sinh (2 k\delta) - \dfrac{\delta}{2}\Big),
\]
then
\[
\big|\psi(0)\big|^2=\big(\sinh^2 (k\delta) \big)\Big(\dfrac{\sinh (2 k\delta)}{4k}-\dfrac{\delta}{2}\Big)^{-1}= 2k+\cO(1)=2\alpha+\cO(1). \qedhere
\]
\end{proof}

\begin{lemma}\label{lem1dr}
Let $\delta>0$ and $\beta\ge 0$ be fixed.
For $\alpha>0$, let $S'$ be the self-adjoint operator in $L^2(0,\delta)$ given by
\[
S'[f,f]=\int_0^\delta |f'|^2\dd t -\alpha \big|f(0)\big|^2 - \beta \big|f(\delta)\big|^2,
\quad
\qdom(S')=H^1(0,\delta),
\]
then for $\alpha\to+\infty$ one has $E_1(S')=-\alpha^2+\cO(e^{-\delta\alpha})$.
Furthermore, there exist $b_\pm>0$ and $b>0$ such that
\begin{equation}
   \label{ejb}
b^- j^2-b \le E_j(S')\le b^+ j^2 \text{ for all $j\ge 2$ and $\alpha\in\RR$}.
\end{equation}
\end{lemma}

\begin{proof}
The operator $S'$ clearly acts as $f\mapsto -f''$
on the functions $f\in H^2(0,\delta)$ with $f'(0)+\alpha f(0)=f'(\delta)-\beta f(\delta)=0$.
To estimate $E_1(S')$ we remark that a value $E=-k^2$ with $k>0$
is an eigenvalue of $S'$ iff one can find  $(C_1,C_2)\in\CC^2\setminus\big\{(0,0)\big\}$ such that the function
$f:t\mapsto C_1e^{kt}+C_2 e^{-kt}$ belongs to its domain. The boundary conditions give
\begin{align*}
0&=f'(0)+\alpha f(0)=(\alpha+k)C_1 + (\alpha-k)C_2,\\
0&=f'(\delta)-\beta f(\delta)=(k-\beta)e^{k\delta} C_1 - (k+\beta)e^{-k\delta}C_2,
\end{align*}
and one has a non-zero solution iff the determinant of the system vanishes, i.e.
iff $k$ solves $(k+\alpha)(k+\beta)e^{-k\delta}=(k-\alpha)(k-\beta)e^{k\delta}$, which we rewrite
as
\begin{equation}
    \label{eq-gh}
g(k)=h(k), \quad g(k):=\dfrac{k+\alpha}{k-\alpha}, \quad h(k):=\dfrac{k-\beta}{k+\beta} \,e^{2k\delta}.
\end{equation}
Both $g$ and $h$ are continuous, and $g$ is strictly decreasing
on $(\alpha,+\infty)$ with $g(\alpha^+)=+\infty$ and $g(+\infty)$=1, while
$h$ is strictly increasing on $(\alpha,+\infty)$ being the product of two strictly increasing positive
functions (we assume without loss of generality that $\alpha>\beta$), and
$h(\alpha^+)=e^{2\alpha\delta}(\alpha-\beta)/(\alpha+\beta) <+\infty$
and $h(+\infty)=+\infty$. Therefore, there exists a unique solution $k$ of \eqref{eq-gh} with $k\in(\alpha,+\infty)$.
To obtain the required estimate we use again the monotonicity of $h$ on $(\alpha,+\infty)$:
\[
\dfrac{k+\alpha}{k-\alpha}=g(k)=h(k)> h(\alpha^+)=\dfrac{\alpha-\beta}{\alpha+\beta} \,e^{2\alpha\delta}.
\]
We bound the last term from below very roughly by $e^{3\alpha\delta/2}$ then
\[
\dfrac{k+\alpha}{k-\alpha}\ge e^{3\alpha\delta/2}, \quad k\le \alpha\,\dfrac{1+e^{-3\alpha\delta/2}}{1-e^{-3\alpha \delta/2}}=
\alpha\big(1+\cO(e^{-3\alpha\delta/2})\big).
\]
By combining with $k>\alpha$ we arrive at the sought estimate 
\[
E_1(S')=-k^2=-\alpha^2\big(1+\cO(e^{-3\alpha\delta/2})\big)=-\alpha^2+\alpha^2\cO(e^{-3\alpha\delta/2})
=-\alpha^2+\cO(e^{-\alpha\delta}).
\]
To estimate $E_j(S')$ with $j\ge 2$ we remark that by the min-max principle
for any $\alpha\in\RR$ one has $E_{j-1}(S'_N)\le E_j(S')\le E_j(S'_D)$, where
the operator $S'_{D/N}$ acts in $L^2(0,\delta)$ as $f\mapsto -f''$
on the functions $f\in H^2(0,\delta)$ with the Dirichlet/Neumann boundary condition at $0$
and $f'(\delta)-\beta f(\delta)$. As the eigenvalues of both $S'_{D/N}$
satisfy the Weyl asymptotics $E_j(S'_{D/N})\sim \pi^2j^2/\delta^2$ as $j\to +\infty$, one arrives at the
inequalities.
\end{proof}

\subsection{Tubular coordinates}\label{sec-curv}

Recall that the shape operator $W$ and curvatures of $\Sigma$ were defined in Subsection~\ref{quad1}.
In what follows we will actively use tubular coordinates on both sides of $\Sigma$. In this section,
\[
\text{let $\Omega_*$ be either $\Omega$ or $\Omega^c$,}
\]
and let $\nu_*$ be the unit normal on $\Sigma$ pointing to the exterior of $\Omega_*$, i.e.
\[
\nu_*:=\nu,\  W_*:=W \text{ for $\Omega_*=\Omega$}, \quad \nu_*:=-\nu,\  W_*:=-W \text{ for $\Omega_*=\Omega^c$}.
\]
The principal curvatures and the (higher) mean curvatures of $\Sigma$ with respect to $\nu_*$ will be denoted
by $h^*_j$ and $H^*_k$ respectively, i.e.
\begin{gather*}
\text{$h_j^*:=h_j$ and $H^*_k=H_k$ for $\Omega_*=\Omega$},\\
\text{$h_j^*:=-h_j$ and $H^*_k=(-1)^k H_k$ for $\Omega_*=\Omega^c$.}
\end{gather*}
For small $\delta>0$ denote
\[
\Pi_\delta:=\Sigma\times(0,\delta), \quad
\Omega_*^\delta=\big\{x\in \Omega_*: \dist(x,\Sigma)<\delta\big\}
\]
It is a well known result in differential geometry that there exists a small $\delta_0>0$
such that for sufficiently small $\delta>0$ the map
\[
\Phi_*: \Pi_\delta\to \Omega_*^\delta,
\quad 
(s,t)\mapsto s-t\nu_*(s),
\]
is a diffeomorphism, and $\dist\big(\Phi_*(s,t),\partial U\big)=t$ for $(s,t)\in\Pi_\delta$.
Consider the associated unitary map
\[
\Theta_\delta: L^2(\Omega_*^\delta) \to L^2(\Pi_\delta),
\quad u\mapsto \sqrt{\det (\Phi_*')}\, u\circ \Phi_*
\]
We will use several times the following computations:
\begin{lemma}\label{lem7}
For $\gamma\in\RR$ denote
\[
J_\gamma(u)\equiv J(u):=\int_{\Omega_*^\delta} |\nabla u|^2 \dd x +\int_{\Sigma} \Big(\gamma+\dfrac{H^*_1}{2}\Big) |u|^2\dd s,
\quad u\in H^1(\Omega_*^\delta).
\]
There exist $\delta_0>0$ and $c>0$ such that for any $\gamma\in\RR$ and $\delta\in(0,\delta_0)$
the following assertions hold true with  $v:=\Theta_\delta u$:
\begin{itemize}
\item[(a)] for any $u\in H^1(\Omega_*^\delta)$ with $u=0$ on $\partial \Omega_*^\delta\setminus\Sigma$ one has
\[
J(u)\le \int_{\Pi_\delta} \bigg[(1+c\delta) |\nabla_s v|^2 + |\partial_t v|^2 + \Big(H^*_2 -\dfrac{(H^*_1)^2}{4}+c\delta\Big)
|v|^2 \bigg]\dd s\dd t +\gamma\int_{\Sigma} \big|v(s,0)\big|^2\dd s,
\]
\item[(b)] for any $u\in H^1(\Omega_*^\delta)$ one has
\begin{multline*}
J(u)\ge \int_{\Pi_\delta} \bigg[(1-c\delta) |\nabla_s v|^2 + |\partial_t v|^2 + \Big(H^*_2 -\dfrac{(H^*_1)^2}{4}-c\delta\Big)
|v|^2 \bigg]\dd s\dd t\\
 +\gamma\int_{\partial U} \big|v(s,0)\big|^2\dd s-c\int_{\Sigma} \big|v(s,\delta)\big|^2\dd s,
\end{multline*}
\end{itemize}
where $\nabla_s$ is the gradient on $\Sigma$, i.e. with respect to the coordinates $s\in\Sigma$.
\end{lemma}

\begin{proof}
The metric $G$ on $\Pi_\delta$ induced by the map $\Phi_*$
is given by $G=g\circ (1- tW_*) + \dd t^2$, with $g$ being
the metric on $\Sigma$ induced by the embedding in $\RR^n$, and the volume form is
$\det G \dd s \dd t=\varphi \dd s \dd t$ with $\dd s$ being the volume form on $\Sigma$
and the weight
\begin{equation}
  \label{eq-varphi}
\varphi(s,t)=\prod\nolimits_{j=1}^{n-1}\big(1-t h^*_j(s)\big)=1+\sum\nolimits_{j\ge 1}(-t)^j H^*_{j}(s).
\end{equation}
Denote $w:=u\circ \Phi_*$, then the standard change of variables gives, for any $u\in H^1(\Omega_*^\delta)$,
\[
J(u)=\int_{\Pi_\delta} |\nabla w|^2\varphi\dd s \dd t
+\int_\Sigma\Big(\gamma+\dfrac{H_1^*}{2}\Big)\, \big|w(s,0)\big|^2\dd s,
\]
and we remark that the condition $u=0$ on $\partial \Omega_*^\delta\setminus\Sigma$ is equivalent to $w(\cdot,\delta)=0$.
Due to the above representation of the metric $G$, for a suitable fixed $c_0>0$ one can estimate, uniformly in $u$,
\[
(1-c_0\delta)|\nabla_s w|^2 + |\partial_t w|^2\le
|\nabla w|^2
\le (1+c_0\delta)|\nabla_s w|^2 + |\partial_t w|^2,
\]
with $\nabla_s$ being the gradient on $\Sigma$ (i.e. with respect to the variable $s$), which gives
\begin{multline}
  \label{eq15}
\int_{\Pi_\delta} \big((1-c_0\delta)|\nabla_s w|^2+|\partial_t w|^2\big)\varphi\dd s \dd t
+\int_\Sigma\Big(\gamma+\dfrac{H_1^*}{2}\Big)\, \big|w(s,0)\big|^2\dd s\\
\le J(u)
\le
\int_{\Pi_\delta} \big((1+c_0\delta)|\nabla_s w|^2+|\partial_t w|^2\big)\varphi\dd s \dd t
+\int_\Sigma\Big(\gamma+\dfrac{H_1^*}{2}\Big)\, \big|w(s,0)\big|^2\dd s.
\end{multline}
Recall that $w=\varphi^{-\frac{1}{2}} v$, and that $\varphi=1$ on $\Sigma$. Hence,
\[
\Big(\gamma+\dfrac{H_1^*}{2}\Big)\, \big|w(s,0)\big|^2=
\Big(\gamma+\dfrac{H_1^*}{2}\Big)\, \big|v(s,0)\big|^2,
\]
which allows to transform the last summand in \eqref{eq15}. In addition,
\[
|\nabla_s w|^2\varphi=\Big|\nabla_s v -\dfrac{1}{2\varphi}v \nabla_s \varphi\Big|^2=
|\nabla_s v|^2 +\dfrac{|v|^2}{4\varphi^2} |\nabla_s\varphi|^2 - \dfrac{1}{\varphi}
\Re \big(\langle \nabla_s v, v\nabla_s\varphi\rangle\big).
\]
The Cauchy-Schwarz inequality gives
$\big|\Re \langle \nabla_s v, v\nabla_s\varphi\rangle\big|\le \delta |\nabla_s v|^2+ |v|^2|\nabla_s \varphi|^2/\delta$,
and in view of the expression \eqref{eq-varphi}
for $\varphi$ one has $|\nabla_s \varphi|^2\le c_1\delta^2$ for some $c_1>0$ and all $t\in(0,\delta)$. Therefore, for a suitable $c_2>0$
one estimates, uniformly in $u$,
\[
(1-c_2\delta)|\nabla_s v|^2-c_2\delta |v|^2
\le (1\pm c_0\delta)|\nabla_s w|^2\varphi
\le 
(1+c_2\delta)|\nabla_s v|^2+c_2 \delta |v|^2.
\]
We represent now
\[
|\partial_t w|^2\varphi=\Big|\partial_t v -\dfrac{1}{2\varphi}v \,\partial_t \varphi\Big|^2
=|\partial_t v|^2 - \dfrac{\partial_t \varphi}{2\varphi} \partial_t \big(|v|^2\big)+\dfrac{(\partial_t \varphi)^2}{4\varphi^2} |v|^2
\]
and performing an integration by parts with respect to $t$ in the middle term we have
\begin{multline*}
\int_{\Pi_\delta} |\partial_t w|^2\varphi\dd s\dd t
=\int_{\Pi_\delta} \bigg(
|\partial_t v|^2+ \Big( \partial_t \Big(\dfrac{\partial_t \varphi}{2\varphi}\Big) +\dfrac{(\partial_t \varphi)^2}{4\varphi^2}\Big)
|v|^2
\bigg)\dd s \dd t\\
-\int_\Sigma \dfrac{H_1^*}{2} \big|v(s,0)\big|^2\dd s -\int_\Sigma \dfrac{(\partial_t\varphi)(s,\delta)}{2\varphi(s,\delta)}
\big|v(s,\delta)\big|^2\dd s,
\end{multline*}
while the last summand vanishes for $v(\cdot,\delta)=0$, i.e. for $u=0$ on $\partial \Omega_*^\delta\setminus\Sigma$.
Putting the above estimates together we obtain
\begin{align*}
J(u)&\le\int_{\Pi_\delta} \bigg((1+c_2\delta)|\nabla_s v|^2+|\partial_t v|^2
 +\Big(\dfrac{\partial^2_t \varphi}{2\varphi} -\dfrac{(\partial_t\varphi)^2}{4\varphi^2}+c_2\delta\Big)|v|^2\dd s\dd t\\
&\quad+\gamma\int_\Sigma\big|v(s,\delta)\big|^2\dd s, \quad
u\in H^1(\Omega_*^\delta), \quad u=0 \text{ on } \partial \Omega_*^\delta\setminus\Sigma,\\
J(u)&\ge \int_{\Pi_\delta} \bigg((1-c_2\delta)|\nabla_s v|^2+|\partial_t v|^2
 +\Big(\dfrac{\partial^2_t \varphi}{2\varphi} -\dfrac{(\partial_t\varphi)^2}{4\varphi^2}-c_2\delta\Big)|v|^2\dd s\dd t\\
&\quad+\gamma\int_\Sigma\big|v(s,0)\big|^2\dd s -\int_\Sigma \dfrac{(\partial_t\varphi)(s,\delta)}{2\varphi(s,\delta)}
\big|v(s,\delta)\big|^2\dd s, \quad u\in H^1(\Omega_*^\delta).
\end{align*}
It remains to estimate, with a suitable $c_3>0$,
\[
\Big\|\dfrac{(\partial_t\varphi)(\cdot,\delta)}{2\varphi(\cdot,\delta)}\Big\|_{L^\infty(\Sigma)}\le c_3,
\quad
\Big\|
\dfrac{\partial^2_t \varphi}{2\varphi} -\dfrac{(\partial_t\varphi)^2}{4\varphi^2} -\Big(H^*_2-\dfrac{(H^*_1)^2}{4}\Big)\Big\|_{L^\infty(\Sigma)}\le c_3 \delta
\]
and to choose $c:=\max\{c_2,c_3\}$.
\end{proof}

\section{Proof of Theorem~\ref{thm1a}}\label{sec-thm1}

We are going to show that $E_j(A_m^2)\to E_j(\Dsl^2)$ for each $j\in \NN$ as $m\to-\infty$.
Due to Lemma~\ref{lemld} for each $j\in\NN$ there holds $E_j(\Dsl^2)=E_j(L)$, hence, it is sufficient
to prove that
\begin{equation} \label{conv01}
E_j(L)=\lim_{m\to-\infty} E_j(A_m^2)\text{ for each $j\in\NN$}.
\end{equation}

\subsection{Dirichlet-Neumann bracketing}

For small $\delta>0$ denote
$\Omega_\delta:=\big\{x\in\Omega:\, \dist(x,\Sigma)<\delta\big\}$ and $\Pi_\delta:=\Sigma\times(0,\delta)$
and consider the diffeomorphisms $\Phi:\Pi_\delta\to \Omega_\delta$ given by $(s,t)\mapsto s-t\nu(s)$
together with the associated unitary maps
$\Theta_\delta: L^2(\Omega_\delta,\CC^N)\to L^2(\Pi_\delta,\CC^N)$, 
$\Theta_\delta u= \sqrt{\det (\Phi')}\,\, u\circ\Phi$.

Consider the self-adjoint operator $Z^+_m$ in $L^2(\Omega_\delta,\CC^N)$ given by
\begin{align}
  \label{loc-5}
Z^+_m[u,  u]& =\int_{\Omega_\delta} \big(|\nabla u|^2 +m^2|u|^2\big)\dd x +\int_\Sigma\Big(m+ \dfrac{H_1}{2}\Big)\, |u|^2\dd s,\\
\qdom(Z^+_m)&=\big\{u\in H^1(\Omega_\delta,\CC^N): u=\cB u \text{ on } \Sigma,
\quad u=0 \text{ on } \partial\Omega_\delta\setminus\Sigma\big\},\nonumber
\end{align}
the self-adjoint operator $Z^-_m$ in $L^2(\Omega_\delta,\CC^N)$ given by
\begin{align}
  \label{loc-6}
Z^-_m[u,  u]& =\int_{\Omega_\delta} \big(|\nabla u|^2 +m^2|u|^2\big)\dd x +\int_\Sigma\Big(m+ \dfrac{H_1}{2}\Big)\, |u|^2\dd s,\\
\qdom(Z^-_m)&=\big\{u\in H^1(\Omega_\delta,\CC^N): u=\cB u \text{ on } \Sigma\big\}, \nonumber
\end{align}
and the self-adjoint operator $Z'_m$ in $L^2(\Omega^c_\delta,\CC^N)$ given by
\[
Z'_m[u,  u] =\int_{\Omega^c_\delta} \big(|\nabla u|^2 +m^2|u|^2\big)\dd x, \quad
\qdom(Z'_m)=H^1(\Omega^c_\delta,\CC^N),
\]
with $\Omega_\delta^c := \Omega \setminus \overline{\Omega_\delta}$.
Due to the min-max principle for any $j\in\NN$ and we have the eigenvalue
inequality
\[
E_j(Z^-_m\oplus Z'_m)\le E_j(A_m^2)\le E_j(Z^+_m).
\]
(It is sufficient to apply Proposition~\ref{prop-incl}:
for the left inequality one takes
$T=A_m^2$, $T':=Z^-_m\oplus Z'_m$,
and $J:L^2(\Omega,\CC^N)\mapsto (f^-,f')\in L^2(\Omega_\delta,\Omega^c_\delta)$ defined by $f^-:=f|_{\Omega_\delta}$
and $f':=f|_{\Omega^c_\delta}$, while for the right inequality 
one takes $T:=Z^+_m$, $T':=A_m^2$ and $J:L^2(\Omega_\delta)\to L^2(\Omega)$  the extension by zero.)
Noting that $Z'_m\ge m^2$ we deduce that
\begin{equation}
 \label{loc-8}
E_j(Z^-_m)\le E_j(A_m^2)\le E_j(Z^+_m) \text{ for any $j\in\NN$ with $E_j(Z^+_m)< m^2$}.
\end{equation}
Using the change of coordinates of Lemma~\ref{lem7} to bound $Z^\pm_m[\Theta_\delta^*v,\Theta_\delta^*v]$ from above and below we then obtain
\[
E_j(Z^+_m)\le E_j(Y^+_m), \quad E_j(Z^-_m)\ge E_j(Y^-_m) \text{ for any $j\in\NN$}
\]
with $Y^\pm_m$ being the self-adjoint operators in $L^2(\Pi_\delta,\CC^N)$ given by
\begin{align*}
Y^+_m[v,v]&=\int_{\Pi_\delta} \bigg[(1+c\delta) |\nabla_s v|^2 + |\partial_t v|^2 + \Big(m^2+H_2 -\dfrac{H_1^2}{4}+c\delta\Big)
|v|^2 \bigg]\dd s\dd t\\
&\quad  +m\int_{\Sigma} \big|v(s,0)\big|^2\dd s,\\
\qdom(Y^+_m)&=\big\{v\in H^1(\Pi_\delta,\CC^N): v(\cdot ,0)=\cB v(\cdot,0) \text{ and } v(\cdot,\delta)=0\big\},\\
Y^-_m[v,v]&=\int_{\Pi_\delta} \bigg[(1-c\delta) |\nabla_s v|^2 + |\partial_t v|^2 + \Big(m^2+H_2 -\dfrac{H_1^2}{4}-c\delta\Big)
|v|^2 \bigg]\dd s\dd t\\
&\quad+m\int_{\Sigma} \big|v(s,0)\big|^2\dd s-c \int_{\Sigma} \big|v(s,\delta)\big|^2\dd s,\\
\qdom(Y^-_m)&=\big\{v\in H^1(\Pi_\delta,\CC^N): v(\cdot ,0)=\cB v(\cdot,0)\big\},
\end{align*}
where $c$ is independent of $\delta\in(0,\delta_0)$ and $m\in \RR$ is arbitrary. Therefore, we arrive at the two-sided estimate
\begin{equation}
 \label{loc-9}
E_j(Y^-_m)\le E_j(A_m^2)\le E_j(Y^+_m) \text{ for any $j\in\NN$ with $E_j(Y^+_m)< m^2$}.
\end{equation}

\subsection{Upper bound}
To obtain an upper bound for the eigenvalues of $Y^+_m$ let us consider the self-adjoint operator $S$ in $L^2(0,\delta)$ with
\[
S[f,f]=\int_{0}^\delta |f'|^2\dd t + m\big|f(0)\big|^2, \quad \qdom(S)=\big\{f\in H^1(0,\delta):\, f(\delta)=0\big\}
\]
and let $\psi$ be an eigenfunction for the first eigenvalue normalized by $\|\psi\|^2_{L^2(0,\delta)}=1$. The analysis of Lemma~\ref{lem1dd} shows that
for some $b>0$ one has $E_1(S)\le-m^2+be^{-\delta|m|}$ as $(-m)$ is large, and then
$S[\psi,\psi] + m^2\le be^{-\delta|m|}$.

Let $c>0$ be the same as in the above expressions for $Y^\pm_m$. For small $a\in\RR$, let $L_a$ be the self-adjoint operator in $\cH$ given by 
\begin{equation}
  \label{lavv}
	\begin{aligned}
L_a[g,g]&=\int_\Sigma \Big[(1+ca)|\nabla g|^2 +\Big(H_2-\dfrac{H_1^2}{4}+ca\Big) |g|^2\Big]\dd s,\\
\qdom(L_a)&=H^1(\Sigma,\CC^N)\cap \cH.
\end{aligned}
\end{equation}
Remark that for $a=0$ we recover exactly the operator $L$ and that due to the min-max principle one has
\begin{equation}
  \label{conv03}
	E_j(L)=\lim_{a\to 0} E_j(L_a) \text{ for each $j\in\NN$.}
\end{equation}
Let $j\in\NN$ be fixed and  $g_1,\dots,g_j$ be linearly independent eigenfunctions of $L_\delta$ for the first $j$ eigenvalues, then
the subspace $G:=\vspan(g_1,\dots,g_j)$ is $j$-dimensional and $L_\delta[g,g]/\|g\|^2_\cH\le E_j(L_\delta)$ for any $0\ne g\in G$.
Consider the subspace
\[
V=\{v \in L^2(\Pi_\delta,\CC^N): \,v(s,t)=g(s)\psi(t), \  g\in G\}\subset \qdom (Y^+_m),
\]
then for $v\in V$ with $v(s,t)=g(s)\psi(t)$ and $g\in G$ one has $\|v\|^2_{L^2(\Pi_\delta,\CC^N)}=\|g\|^2_\cH$
and
\begin{multline*}
Y^+_m[v,v]=L_\delta[g,g] \|\psi\|^2_{L^2(0,\delta)} + \Big(S[\psi,\psi] + m^2\|\psi\|^2_{L^2(0,\delta)}\Big)\|g\|^2_\cH\\
\le L_\delta[g,g] + be^{-\delta|m|}\|g\|^2_\cH\le \big( E_j(L_\delta) + be^{-\delta|m|}\big)\|g\|^2_\cH\\
\equiv \big( E_j(L_\delta) + be^{-\delta|m|}\big)\|v\|^2_{L^2(\Pi_\delta,\CC^N)}.
\end{multline*}
As $\dim V=\dim G=j$, it follows by the min-max principle that
\[
E_j(Y^+_m)\le \sup_{0\ne v\in V}\dfrac{Y^+_m[v,v]}{\|v\|^2_{L^2(\Pi_\delta,\CC^N)}}\le E_j(L_\delta) + be^{-\delta|m|},
\]
hence, $\limsup_{m\to-\infty}E_j(Y^+_m)\le E_j(L_\delta)$. As $\delta>0$ can be chosen arbitrarily small, the convergence \eqref{conv03}
implies  $\limsup_{m\to-\infty}E_j(Y^+_m)\le E_j(L)$, and then due to the upper bound \eqref{loc-9} we arrive at
\begin{equation}
  \label{conv04}
	\limsup_{m\to-\infty}E_j(A_m^2)\le E_j(L).
\end{equation}

\subsection{Lower bound}
Now let us pass to a lower bound for $E_j(Y^-_m)$. In the constructions below, the constant $c>0$ is the same as in the expression for $Y^-_m$.
Let $S'$ be the self-adjoint operator in $L^2(0,\delta)$ with
\[
S'[f,f]=\int_0^\delta |f'|^2\dd t +m \big|f(0)\big|^2 - c \big|f(\delta)\big|^2,
\quad
\qdom (S')=H^1(0,\delta).
\]
Let $\psi_k\in L^2(0,\delta)$ with $k\in\NN$ be real-valued eigenfunctions of $S'$ for the eigenvalues $E_k(S')$ forming an orthonormal basis in $L^2(0,\delta)$, 
which induces the unitary transforms $\Theta:L^2(0,\delta)\to \ell^2(\NN)$ given by $(\Theta f)_k=\langle \psi_k,f\rangle_{L^2(0,\delta)}$, $k\in\NN$.
Recall that due to the analysis of Lemma~\ref{lem1dr} we have, with some $b^\pm>0$, $b>0$ and $b_0>0$,
\begin{gather}
    \label{conv05}
E_1(S')\ge -m^2-b e^{-\delta|m|} \text{ as } m\to-\infty, \\
    \label{conv06}
b^- k^2- b_0\le E_k(S')\le b^+ k^2 \text{ for all $k\ge 2$ and $m\in \RR$.}
\end{gather}

Let us give some more details on the subsequent constructions. Let $Y_m$ be the self-adjoint
operator whose sesquilinear form is given by the same expression as the one for $Y^-_m$ but on the larger form domain $\qdom(Y_m)=H^1(\Pi_\delta,\CC^N)$.
It follows easily that the new operator $Y_m$ admits a separation of variables. Namely, for small $a\in\RR$
we consider the self-adjoint operator $\Lambda_a$ in $L^2(\Sigma,\CC^N)$ given by
\[
\Lambda_a[g,g]=\int_\Sigma \Big[(1+ca)|\nabla g|^2 +\Big(H_2-\dfrac{H_1^2}{4}+ca\Big) |g|^2\Big]\dd s,
\quad
\qdom(\Lambda_a)=H^1(\Sigma,\CC^N),
\]
i.e. its sesquilinear form is given by the same expression as the one for $L_a$ in \eqref{lavv}
but without the restriction $g\in\cH$.
Now, if one identifies $L^2(\Pi_\delta,\CC^N)=L^2(0,\delta)\otimes L^2(\Sigma,\CC^N)$, then
$Y_m=(S'+m^2)\otimes 1 + 1\otimes \Lambda_{-\delta}$. Using the unitary transform
\begin{gather*}
\Xi:L^2(\Pi_\delta)\to \ell^2(\NN)\otimes L^2(\Sigma,\CC^N),\\
\Xi v=(v_k), \quad v_k:=\int_0^\delta \psi_k(t)v(t,\cdot)\dd t\in L^2(\Sigma,\CC^N),
\end{gather*}
and the spectral theorem we see that the operator $\widehat Y_m:=\Xi Y_m\Xi^*$ is given by
\[
\widehat Y_m\big[(v_k),(v_k)\big]=\sum_{k\in\NN} \Big( \Lambda_{-\delta}[v_k,v_k] +\big(E_k(S')+m^2\big)\|v_k\|^2_{L^2(\Sigma,\CC^N)}\Big),
\]
while  the form domain $\qdom(\widehat Y_m)$ consists of all $(v_k)\in \ell^2(\NN)\otimes L^2(\Sigma,\CC^N)$ with
$v_k\in H^1(\Sigma,\CC^N)$ such that the right-hand side of the preceding expression is finite.
Using the two-sided estimate \eqref{conv06} we can rewrite
\begin{multline}
\qdom(\widehat Y_m)=\Big\{ (v_k)\in \ell^2(\NN)\otimes L^2(\Sigma,\CC^N):\ 
v_k\in H^1(\Sigma,\CC^N) \text{ for each $k\in \NN$}\\
\text{and } \sum_{k\in\NN} \Big(\|v_k\|^2_{H^1(\Sigma,\CC^N)}+k^2\|v_k\|^2_{L^2(\Sigma,\CC^N)}\Big)<\infty\Big\}.
   \label{qdomy0}
\end{multline}
As the sesquilinear form for $Y^-_m$ is simply the restriction of that for $Y_m$ on the functions $v$
with  $v(\cdot,0)=\cB v(\cdot,0)$, for the operator $\widehat Y^-_m:=\Xi Y^-_m \Xi^*$ we have
\begin{equation}
   \label{qdomy1}
\qdom(\widehat Y^-_m)=\big\{ \widehat v=(v_k)\in \qdom(\widehat Y_m): (1-\cB) (\Xi^*\widehat v)(\cdot,0)=0\big\}.
\end{equation}
Using the lower bounds \eqref{conv05} and \eqref{conv06} for $E_k(S')$, for all $\widehat v=(v_k)\in \qdom(\widehat Y^-_m)$
we obtain the inequality $\widehat Y^-_m[\widehat v,\widehat v]\ge w_m(\widehat v,\widehat v)$
with the sesquilinear form $w_m$ defined on $\dom (w_m):=\qdom(\widehat Y^-_m)$ by
\begin{multline*}
w_m(\widehat v,\widehat v):=\Lambda_{-\delta}[v_1,v_1]-b e^{-\delta|m|}\|v_1\|^2_{L^2(\Sigma,\CC^N)}\\
+\sum_{k\ge 2} \Big( \Lambda_{-\delta}[v_k,v_k] + (b^-k^2-b_0 +m^2)\|v_k\|^2_{L^2(\Sigma,\CC^N)}\Big).
\end{multline*}
It follows from representation \eqref{qdomy0} that
the form $w_m$  is lower semibounded and from reprentation \eqref{qdomy1} that it is closed. Thus, it defines a self-adjoint operator $W_m$
in $\ell^2(\NN)\otimes L^2(\Sigma,\CC^N)$ with compact resolvent. For any $j\in\NN$ we have then
\begin{equation}
   \label{conv07a}
E_j(A_m^2)\ge E_j(Y^-_m)=E_j(\widehat Y^-_m)\ge E_j(W_m).
\end{equation}
We are now in the classical situation for the monotone convergence (Proposition~\ref{prop-mon}) to analyze the eigenvalues of $W_m$.
Namely, consider the set
\begin{equation}
  \label{conv08}
\cQ_\infty:=\Big\{\widehat v=(v_k)\in \bigcap_{m<0} \qdom (W_m)\equiv \qdom(\widehat Y^-_m):\quad \sup_{m<0} W_m[\widehat v,\widehat v]<+\infty\Big\}.
\end{equation}
It is easily seen that a vector $\widehat v=(v_k)\in \qdom(\widehat Y^-_m)$ belongs to $\cQ_\infty$ if and only if
$v_k=0$ for $k\ge 2$ and $0=(1-\cB) (\Xi^*\widehat v)(\cdot,0)\equiv\psi_1(0)(1-\cB)v_1 $, i.e. $v_1\in\cH$. This gives the equality
\[
\cQ_\infty=\big\{ \widehat v=e_1\otimes v_1:\, v_1\in H^1(\Sigma,\CC^N)\cap \cH\}, \quad e_1=(1,0,0,\dots)\in \ell^2(\NN).
\]
For each $\widehat v\in \cQ_\infty$ one has
\begin{align*}
\lim_{m\to-\infty} W_m[\widehat v,\widehat v]&=\lim_{m\to-\infty}  \Big(\Lambda_{-\delta}[v_1,v_1]-c_1 e^{-\delta|m|}\|v_1\|^2_{L^2(\Sigma,\CC^N)}\Big)\\
&=\Lambda_{-\delta}[v_1,v_1]\equiv L_{-\delta}[v_1,v_1];
\end{align*}
we recall that $L_a$ was defined in \eqref{lavv}.
Let $W_\infty$ be the self-adjoint operator in the Hilbert space $\cH_\infty:=e_1\otimes \cH$ with
$\qdom(W_\infty)=\cQ_\infty$ and $W_\infty[e_1\otimes v_1,e_1\otimes v_1]=L_{-\delta}[v_1,v_1]$,
then the monotone convergence principle (Proposition~\ref{prop-mon})
gives $\lim_{m\to-\infty} E_j(W_m)=E_j(W_\infty)$ for each $j\in\NN$. On the other hand,
the operator $W_\infty$ is unitarily equivalent to $L_{-\delta}$, and by combining with \eqref{conv07a} we have
$\liminf_{m\to-\infty}E_j(A_m)\ge E_j(L_{-\delta})$. As $\delta$ can be arbitrarily small, the convergence \eqref{conv03}
implies $\liminf_{m\to-\infty}E_j(A_m)\ge E_j(L)$. In combination with the upper bound \eqref{conv04}
one arrives at the sought limit \eqref{conv01}, which proves Theorem~\ref{thm1a}.

\section{Proof of Theorem~\ref{thm2}}\label{sec-thm2}

\subsection{Preliminary estimates}
We are going to prove that for each $m\in\RR$ and $j\in\NN$ one has
$\lim_{M\to+\infty}E_j(B_{m,M}^2)=E_j(A_m^2)$. We recall that $\qdom(B_{m,M}^2)\equiv \dom(B_{m,M})=H^1(\RR^n,\CC^N)$, and
\begin{equation}
  \label{qform3}
\begin{aligned}	
B_{m,M}^2[u, u]&\equiv \langle B_{m,M} u, B_{m,M} u\rangle_{L^2(\RR^n,\CC^N)}\\
&=\int_{\Omega} \big(|\nabla u|^2 +m^2|u|^2\big)\dd x
+ \int_{\Omega^c} \big(|\nabla u|^2 +M^2|u|^2\big)\dd x\\
&\qquad+(M-m)\Big(\int_\Sigma | \cP_- u|^2\dd s -\int_\Sigma | \cP_+ u|^2\dd s\Big),\\
&=\int_{\Omega} \big(|\nabla u|^2 +m^2|u|^2\big)\dd x 
+ \int_{\Omega^c} \big(|\nabla u|^2 +M^2|u|^2\big)\dd x\\
&\qquad+2(M-m)\int_\Sigma | \cP_- u|^2\dd s + (m-M)\int_\Sigma |u|^2\dd s
\end{aligned}
\end{equation}
where $\cP_\pm(s):=\dfrac{1 \pm \cB(s)}{2}$ for $s\in\Sigma$, while
\begin{gather*}
\qdom(A_{m}^2)\equiv \dom(A_{m})=\big\{u\in H^1(\Omega,\CC^N):\ \cP_- u=0\text{ on }\Sigma \big\},\\
A_m^2[u,u]\equiv\langle A_m u, A_m u\rangle_{L^2(\Omega,\CC^N)}
=\int_{\Omega} \big(|\nabla u|^2 +m^2|u|^2\big)\dd x +\int_\Sigma\Big(m+ \dfrac{H_1}{2}\Big)\, |u|^2\dd s.
\end{gather*}
Taking any $\varepsilon\in\RR$ we rewrite the above expression for $B_{m,M}^2[u,u]$ as
\begin{multline}
B_{m,M}^2[u, u]=\int_{\Omega} \big(|\nabla u|^2 +m^2|u|^2\big)\dd x\\
 +\int_\Sigma \Big(m-\varepsilon+\dfrac{H_1}{2}\Big)|u|^2\dd s + 2(M-m)\int_\Sigma | \cP_- u|^2\dd s\\
 +\int_{\Omega^c} \big(|\nabla u|^2 +M^2|u|^2\big)\dd x -\int_\Sigma \Big(M-\varepsilon+\dfrac{H_1}{2}\Big)|u|^2\dd s.
  \label{eq-bmm1}
\end{multline}
Let us start with an additional estimate which will allow us to  control the term in the last line of \eqref{eq-bmm1}.

\begin{lemma}\label{lem11}
For $\gamma>0$ let $R_\gamma$ be the self-adjoint operator in $L^2(\Omega^c)$ given by
\begin{equation}
R_\gamma[u,u]=\int_{\Omega^c} |\nabla u|^2 \dd x -\int_{\Sigma} \Big(\gamma+\dfrac{H_1}{2}\Big) |u|^2\dd s,
\quad
\qdom(R_\gamma)=H^1(\Omega^c),
\end{equation}
then:
\begin{itemize}
\item[(a)] For some fixed $C>0$ and all large $\gamma>0$ there exists a linear map $F_\gamma: H^1(\Sigma)\to H^1(\Omega^c)$ such that
for all $f\in H^1(\Sigma)$ one has $F_\gamma f = f$ on $\Sigma$ and
\[
R_\gamma[F_\gamma f,F_\gamma f]+\gamma^2\|F_\gamma f\|_{L^2(\Omega^c)}^2\le \dfrac{C}{\gamma}\|f\|^2_{H^1(\Sigma)}.
\]
\item[(b)] For some $C_0>0$ there holds $E_1(R_\gamma)\ge -\gamma^2-C_0$ for $\gamma\to+\infty$.
\end{itemize}
\end{lemma}

\begin{proof}
For a small $\delta>0$ consider the sets $\Omega^c_\delta:=\big\{x\in\Omega^c:\, \dist(x,\Sigma)<\delta\big\}$ and $\Pi_\delta:=\Sigma\times(0,\delta)$
together with the the diffeomorphisms $\Phi^c:\Pi_\delta\to \Omega^c_\delta$ given by $\Phi^c(s,t)\mapsto s+t\nu(s)$
and the associated unitary maps $\Theta^c_\delta: L^2(\Omega^c_\delta)\to L^2(\Pi_\delta)$ with $\Theta^c_\delta u= \sqrt{\det \big((\Phi^c)'\big)}\,\, u\circ \Phi^c$.

Let us prove (a). Consider the self-adjoint operator $S$ in $L^2(0,\delta)$ given by
\[
S[f,f]=\int_{0}^\delta |f'|^2\dd t -\gamma\big|f(0)\big|^2, \quad \qdom(S)=\big\{f\in H^1(0,\delta):\, f(\delta)=0\big\}
\]
and let $\psi$ be an eigenfunction for the first eigenvalue normalized by $\psi(0)=1$. By Lemma~\ref{lem1dd}, with some
$b>0$ one has $E_1(S)\le-\gamma^2+b$ and $\|\psi\|^2_{L^2(0,\delta)}\le b/\gamma$ as $\gamma$ is large.
For $f\in H^1(\Sigma)$ define $v\in H^1(\Pi_\delta)$ by $v=f\otimes\psi$, i.e. $v(s,t)=f(s) \psi(t)$,
and then set
\[
(F_\gamma f)(x):=\begin{cases}
(\Theta^c_\delta)^{-1} v & \text{ in } \Omega^c_\delta,\\
0 &\text{ in } \Omega^c\setminus \Omega^c_\delta.
\end{cases}
\]
Due to $f\in H^1(\Sigma)$ and $\psi(\delta)=0$ one has $F_\gamma f\in H^1(\Omega^c)$, and the equality
$F_\gamma f|_\Sigma=v(\cdot,0)=f$ holds by construction. Furthermore,
using the result and the notation of Lemma~\ref{lem7}(a) we obtain, with some $a>0$,
\begin{multline*}
R_\gamma[F_\gamma f,F_\gamma f] + \gamma^2 \|F_\gamma f\|^2=J_{-\gamma}(F_\gamma f)+ \gamma^2\|F_\gamma f\|^2\\
\begin{aligned}
&\le \int_{\Pi_\delta}\Big( a|\nabla_s v|^2 + |\partial_t v|^2 + (\gamma^2+a)|v|^2\big)\dd s\dd t -\gamma\int_\Sigma |F_\gamma f|^2\dd s\\
&=\int_{\Sigma} \Big(a|\nabla_s f|^2\dd s + \big(E_1(S) +\gamma^2 +a\big) \|f\|^2_{L^2(\Sigma)}\Big)\dd s\, \|\psi\|^2_{L^2(0,\delta)}\\
&\le \int_{\Sigma}\Big(a|\nabla_s f|^2\dd s + \big(be^{-\delta \gamma} +a\big) \|f\|^2_{L^2(\Sigma)}\Big)\dd s \, \dfrac{b}{\gamma}\le \dfrac{C}{\gamma} \|f\|^2_{H^1(\Sigma)}
\end{aligned}
\end{multline*}
with $C:= b\big(b +a\big)$. Hence, the assertion (a) is proved.

To prove (b) we remark first that due to the min-max principle
one has the inequality $E_1(R_\gamma)\ge E_1(R^0_\gamma\oplus R'_\gamma)$ where $R^0_\gamma$ is the operator in $L^2(\Omega^c_\delta)$ given by
\[
R^0_\gamma[u,u]=\int_{\Omega^c_\delta} |\nabla u|^2 \dd x -\int_{\Sigma} \Big(\gamma+\dfrac{H_1}{2}\Big) |u|^2\dd s,
\quad
\qdom(R^0_\gamma)=H^1(\Omega^c_\delta),
\]
and $R'_\gamma$ is the self-adjoint operator in $L^2(\Omega'_\delta)$, with $\Omega'_\delta:=\Omega^c\setminus \overline{\Omega^c_\delta}$, given by
\[
R'_\gamma[u,u]=\int_{\Omega'_\delta} |\nabla u|^2 \dd x, \quad \qdom(R'_\gamma)=H^1(\Omega_\delta').
\]
Due to $R'_\gamma\ge 0$ one has $E_1(R_\gamma)\ge \min\{E_1(R^0_\gamma),0\}$. By Lemma~\ref{lem7}(b) one has
$E_1(R^0_\gamma)\ge E_1(X_\gamma)$ with $R$ being the self-adjoint operator in $L^2(\Pi_\delta)$
with
\[
X_\gamma[v,v]= \int_{\Pi_\delta} \bigg[a' |\nabla_s v|^2 + |\partial_t v|^2 -a' |v|^2 \bigg]\dd s\dd t\\
 -\gamma\int_{\Sigma} \big|v(s,0)\big|^2\dd s-a'\int_{\partial U} \big|v(s,\delta)\big|^2\dd s
\]
and $\qdom(X_\gamma)=H^1(\Omega^c_\delta)$, with some $a'>0$. Let $S'$ be the self-adjoint operator
in $L^2(0,\delta)$ given by
\[
S'[f,f]=\int_0^\delta |f'|^2\dd t -\gamma \big|f(0)\big|^2 - a' \big|f(\delta)\big|^2,
\quad
\qdom(S')=H^1(0,\delta).
\]
As $|\nabla_s v|^2\ge 0$, due to Fubini's theorem one has $E_1(X_\gamma)\ge E_1(S')-a'$, and
now it is sufficient to remark that by Lemma~\ref{lem1dr} one has $E_1(S')\ge -\gamma^2-a_0$ with some $a_0>0$ as $\gamma\to+\infty$.
\end{proof}

\subsection{Upper bound}

Pick $m\in\RR$ and $j\in\NN$, and let $u_1,\dots,u_j$ be linearly independent eigenfunctions of $A_m^2$ for the first $j$ eigenvalues, then
for any function $u\in V:=\vspan(u_1,\dots,u_j)$ there holds $A_m^2[u,u]\le E_j(A_m^2)\|u\|^2_{L^2(\Omega,\CC^N)}$. Recall that
due to Lemma~\ref{qfa} one has $V\subset C^\infty(\overline{\Omega},\CC^N)$, and then
\[
a:=\sup\big\{ \|u\|^2_{H^1(\Sigma,\CC^N)}: u\in V \text{ with } \|u\|^2_{L^2(\Omega,\CC^N)}=1\big\}<\infty.
\]
Using the linear map $F_\gamma$ as in Lemma~\ref{lem11}(a),
for  $u\in V$ define $\widetilde u \in H^1(\RR^n,\CC^N)$ by
\[
\widetilde u=\begin{cases}
u & \text{ in } \Omega,\\
(F_M\otimes 1)(u|_\Sigma) & \text{ in } \Omega^c.
\end{cases}
\]
with $1$ understood as the identity operator in $\CC^N$, then for any $u\in V$
we have
\begin{multline*}
\int_{\Omega^c} \big(|\nabla \widetilde u|^2 +M^2|\widetilde u|^2\big)\dd x -\int_\Sigma \Big(M+\dfrac{H_1}{2}\Big)| \widetilde u|^2\dd s\\
\equiv \Big((R_M+M^2)\otimes 1\Big)[\widetilde u,\widetilde u]\le \dfrac{C}{M} \|u\|^2_{H^1(\Sigma,\CC^N)}\le \dfrac{C a}{M} \|u\|^2_{L^2(\Omega,\CC^N)}
\end{multline*}
with $C>0$ independent of $u$.
Noting that for $u\in V$ we have $\cP_- u=0$ on $\Sigma$ and substituting the preceding upper bound into \eqref{eq-bmm1} with the choice $\varepsilon=0$
we arrive at
\[
B_{m,M}^2[\widetilde u,\widetilde u]=A_m^2[u,u] +\big((R_M+M^2)\otimes 1\big)[\widetilde u,\widetilde u]
\le \Big(E_j(A_m^2) + \dfrac{C a}{M}\Big)\|u\|^2_{L^2(\Omega,\CC^N)}.
\]
For $u\in V$ there holds $\|\widetilde u\|^2_{L^2(\RR^n,\CC^N)}\ge \|u\|^2_{L^2(\Omega,\CC^N)}$,
and $\widetilde V:=\{\widetilde u:\, u\in V\}$ is therefore a $j$-dimensional subspace of $H^1(\RR^n,\CC^N)\equiv \qdom(B_{m,M}^2)$.
The min-max principle gives
\begin{multline*}
E_j(B_{m,M}^2)\le \sup_{0\ne v\in \widetilde V} \dfrac{B_{m,M}^2[v,v]}{\|v\|^2_{L^2(\RR^n,\CC^N)}}=\sup_{0\ne u\in V} \dfrac{B_{m,M}^2[\widetilde u,\widetilde u]}{\|\widetilde u\|^2_{L^2(\RR^n,\CC^N)}}\\
\le \sup_{0\ne u\in V} \dfrac{\Big(E_j(A_m^2) + \dfrac{C a}{M}\Big)\|u\|^2_{L^2(\Omega,\CC^N)}}{\|\widetilde u\|^2_{L^2(\RR^n,\CC^N)}}\le E_j(A_m^2) + \dfrac{C a}{M},
\end{multline*}
which implies $\limsup_{M\to+\infty}E_j(B_{m,M}^2)=E_j(A_m^2)$.

\subsection{Lower bound}

Now we use the representation \eqref{eq-bmm1} with an arbitrary fixed $\varepsilon>0$.
By the min-max principle, for any $j\in\NN$ one has
\begin{equation}
   \label{ejej0}
E_j(B_{m,M}^2)\ge E_j(K_{m,M,\varepsilon} \oplus K^c_{M,\varepsilon})
\end{equation}
where $K_{m,M,\varepsilon}$ is the self-adjoint operator in $L^2(\Omega,\CC^N)$ with the form
domain $\qdom(K_{m,M,\varepsilon})=H^1(\Omega,\CC^N)$ and
\begin{multline*}
K_{m,M,\varepsilon}[u,u]
=\int_{\Omega} \big(|\nabla u|^2 +m^2|u|^2\big)\dd x\\
+\int_\Sigma \Big(m-\varepsilon+\dfrac{H_1}{2}\Big)|u|^2\dd s + 2(M-m)\int_\Sigma | \cP_- u|^2\dd s,
\end{multline*}
and $K^c_{M,\varepsilon}$ is the self-adjoint operator in $L^2(\Omega^c,\CC^N)$ with the form domain
$\qdom(K^c_{M,\varepsilon})=H^1(\Omega^c,\CC^N)$ and
\[
K^c_{M,\varepsilon}[u,u]=\int_{\Omega^c} \big(|\nabla u|^2 +M^2|u|^2\big)\dd x -\int_\Sigma \Big(M-\varepsilon+\dfrac{H_1}{2}\Big)|u|^2\dd s.
\]
Using the operator $R_\gamma$ from  Lemma~\ref{lem11} one easily
sees that $K^c_{M,\varepsilon}=(R_{M-\varepsilon}\otimes 1)+M^2$ with $1$ being the identity in $\CC^N$, and then, using Lemma~\ref{lem11}(b),
 $E_1(K^c_{M,\varepsilon})=E_1(R_{M-\varepsilon})+M^2\ge \varepsilon M$ as $M$ is large. Due to the upper bound proved in the preceding subsection we know already
that for each fixed $j\in\NN$ there holds $E_j(B_{m,M}^2)=\cO(1)$ for large $M$, hence, Eq.~\eqref{ejej0} implies
\[
E_j(B_{m,M}^2)\ge \min\big\{E_j(K_{m,M,\varepsilon}), E_1(K^c_{M,\varepsilon})\big\}=E_j(K_{m,M,\varepsilon}) \text{ as } M\to+\infty.
\]
As the operators $K_{m,M,\varepsilon}$ are increasing with respect to $M$,
with the help of the monotone convergence (Proposition~\ref{prop-mon}) for each $j\in\NN$
one obtains $\lim_{M\to+\infty} E_j(K_{m,M,\varepsilon})=E_j(C_{m,\varepsilon})$, where
$C_{m,\varepsilon}$ is the self-adjoint operator in $L^2(\Omega,\CC^N)$ given by
\begin{gather*}
C_{m,\varepsilon}[u,u]
=\int_{\Omega} \big(|\nabla u|^2 +m^2|u|^2\big)\dd x
+\int_\Sigma \Big(m-\varepsilon+\dfrac{H_1}{2}\Big)|u|^2\dd s,\\
\qdom(C_{m,\varepsilon})=\big\{
u\in H^1(\Omega,\CC^N):\, \cP_- u=0 \text{ on } \Sigma
\big\}\equiv \qdom (A_m^2).
\end{gather*}
This shows that $\liminf_{M\to+\infty} E_j(B_{m,M}^2)\ge E_j(C_{m,\varepsilon})$. As $\varepsilon>0$ is arbitrary
and we have the obvious limit $\lim_{\varepsilon\to 0} E_j(C_{m,\varepsilon})=E_j(C_{m,0})\equiv E_j(A_m^2)$,
we arrive at the sought lower bound $\liminf_{M\to+\infty} E_j(B_{m,M}^2)\ge E_j(A_m^2)$, which finishes the proof.

\section{Proof of Theorem~\ref{thm3}}\label{sec-thm3}

We are going to show that for each $j\in\NN$ the eigenvalues $E_j(B_{m,M}^2)$ converge to $E_j(\Dsl^2)$ as $m\to -\infty$ and $M\to+\infty$ with $m/M\to 0$.
Due to Lemma~\ref{lemld} for each $j\in\NN$ there holds $E_j(\Dsl^2)=E_j(L)$, hence, it is sufficient
to prove that $E_j(B_{m,M}^2)$ converges to $E_j(L)$ in the same asymptotic regime.
The proof is essentially by combining in a new way some constructions used in the proofs of Theorems~\ref{thm1a} and~\ref{thm2}.

\subsection{Upper bound}

Let us recall the important technical ingredients. For small $\delta>0$ consider the sets
$\Omega_\delta:=\{x\in\Omega:\, \dist(x,\Sigma)<\delta\}$ and $\Pi_\delta:=\Sigma\times(0,\delta)$
as well as the diffeomorphisms $\Phi:\Pi_\delta\to \Omega_\delta$ given by $\Phi(s,t)=s-t\nu(s)$
and the associated unitary maps $\Theta_\delta: L^2(\Omega_\delta,\CC^N)\to L^2(\Pi_\delta,\CC^N)$ with
$\Theta_\delta u= \sqrt{\det(\Phi')}\,\, u\circ\Phi$.

Consider the self-adjoint operator $S$ in $L^2(0,\delta)$ with
\[
S[f,f]=\int_{0}^\delta |f'|^2\dd t + m\big|f(0)\big|^2, \quad \qdom(S)=\big\{f\in H^1(0,\delta):\, f(\delta)=0\big\}
\]
and let $\psi$ be an eigenfunction for the first eigenvalue normalized by $\|\psi\|^2_{L^2(0,\delta)}=1$.
By Lemma~\ref{lem1dd} with some $b>0$ one has
\[
E_1(S)\le-m^2+be^{-\delta|m|},
\quad
\big|\psi(0)\big|^2\le b |m|,
\text{ as $(-m)$ is large}.
\]
Also recall that due to Lemma~\ref{lem11}(a) one can find $c>0$ such that for $\delta\in(0,\delta_0)$
and $u\in H^1(\Omega_\delta)$ with $u=0$ on $\partial \Omega_\delta\setminus\Sigma$ there holds,
with $w:=\Theta_\delta u$,
\begin{multline}
   \label{upp01}
\int_{\Omega_\delta} |\nabla u|^2\dd x +\int_{\partial U} \Big(m+\dfrac{H_1}{2}\Big) |u|^2\dd s\\
\le
\int_{\Pi_\delta} \bigg[(1+c\delta) |\nabla_s w|^2 + |\partial_t w|^2 + \Big(H_2 -\dfrac{H_1^2}{4}+c\delta\Big)
|w|^2 \bigg]\dd s\dd t\\
+m\int_{\Sigma} \big|w(s,0)\big|^2\dd s.
\end{multline}
We will use the representation \eqref{eq-bmm1} with $\varepsilon=0$, i.e.
\begin{multline}
B_{m,M}^2[u, u]=\int_{\Omega} \big(|\nabla u|^2 +m^2|u|^2\big)\dd x +\int_\Sigma \Big(m+\dfrac{H_1}{2}\Big)|u|^2\dd s + 2(M-m)\int_\Sigma | \cP_- u|^2\dd s\\
+ \int_{\Omega^c} \big(|\nabla u|^2 +M^2|u|^2\big)\dd x -\int_\Sigma \Big(M+\dfrac{H_1}{2}\Big)|u|^2\dd s, \quad u\in H^1(\RR^n,\CC^N).
  \label{eq-bmm2}
\end{multline}
For small $a\in\RR$ consider the operator $L_a$ in $\cH$ given by
\begin{equation}
  \label{lavv1}
	\begin{aligned}
L_a[g,g]&=\int_\Sigma \Big[(1+ca)|\nabla g|^2 +\Big(H_2-\dfrac{H_1^2}{4}+ca\Big) |g|^2\Big]\dd s,\\
\qdom(L_a)&=H^1(\Sigma,\CC^N)\cap \cH.
\end{aligned}
\end{equation}
Finally, by Lemma~\ref{lem11} for large $M>0$ there exists $C>0$ and a linear extension map $F_M:H^1(\Sigma,\CC^N)\to H^1(\Omega^c,\CC^N)$
with $(F_M f) |_\Sigma=f$ and
\[
\int_{\Omega^c} \big(|\nabla F_M f|^2 +M^2|F_M f|^2\big)\dd x -\int_\Sigma \Big(M+\dfrac{H_1}{2}\Big)|F_M f|^2\dd s\le \dfrac{C}{M}\, \|f\|^2_{H^1(\Sigma,\CC^N)}.
\]
for all $f\in H^1(\Sigma,\CC^N)$.

Let $j\in\NN$ and $v_1,\dots,v_j$ be linearly independent eigenfunctions of $L_\delta$ for the first $j$ eigenvalues, then for 
$v\in V:=\vspan(v_1,\dots,v_j)$ one has $L_\delta[v,v]\le E_j(L_\delta)\,\|v\|^2_\cH\equiv E_j(L_\delta)\,\|v\|^2_{L^2(\Sigma,\CC^N)}$.
Denote
\[
a_0:=\sup\big\{ \|v\|^2_{H^1(\Sigma,\CC^N)}:\, v\in V \text{ with } \|v\|^2_{\cH}=1\big\}<\infty.
\]
For $v\in V$ construct $u\in H^1(\RR^n,\CC^N)$ as follows:
\[
u=\begin{cases}
\Theta_\delta^{-1} (v\otimes \psi)& \text{ in } \Omega_\delta,\\
\psi(0)F_M v& \text{ in } \Omega^c,\\
0& \text{ in } \Omega\setminus \Omega_\delta.
\end{cases}
\]
By construction one has
\[
\|u\|^2_{L^2(\RR^n,\CC^N)}\ge \|u\|^2_{L^2(\Omega_\delta,\CC^N)}=\|v\|^2_{L^2(\Sigma,\CC^N)} \|\psi\|^2_{L^2(0,\delta)}=\|v\|^2_{L^2(\Sigma,\CC^N)}
\equiv \|v\|^2_\cH,
\]
hence, the subspace $U:=\{u:\, v\in V\}\subset H^1(\RR^n,\CC^N)$ is $j$-dimensional. By the above properties of $F_M$ and $\psi$
one has
\begin{multline*}
\int_{\Omega^c} \big(|\nabla u|^2 +M^2|u|^2\big)\dd x -\int_\Sigma \Big(M+\dfrac{H_1}{2}\Big)|u|^2\dd s\\
=\big|\psi(0)\big|^2 \bigg(\int_{\Omega^c} \big(|\nabla F_M v|^2 +M^2|F_M v|^2\big)\dd x -\int_\Sigma \Big(M+\dfrac{H_1}{2}\Big)|F_M v|^2\dd s\bigg)\\
\le \big|\psi(0)\big|^2 \,\dfrac{C}{M}\, \|v\|^2_{H^1(\Sigma,\CC^N)}\le b|m| \, \dfrac{C}{M}\,a_0\|v\|^2_{L^2(\Sigma,\CC^N)}\equiv a_0bC\dfrac{|m|}{M}\,\|v\|^2_{\cH},
\end{multline*}
and due to \eqref{upp01} there holds
\begin{multline*}
\int_{\Omega} \big(|\nabla u|^2 +m^2|u|^2\big)\dd x +\int_\Sigma \Big(m+\dfrac{H_1}{2}\Big)|u|^2\dd s + 2(M-m)\int_\Sigma | \cP_- u|^2\dd s\\
\begin{aligned}
\equiv
&\int_{\Omega_\delta} \big(|\nabla u|^2 +m^2|u|^2\big)\dd x +\int_\Sigma \Big(m+\dfrac{H_1}{2}\Big)|u|^2\dd s\\
&\le \int_0^\delta\int_{\Sigma} \bigg[(1+c\delta) \big|\nabla_s (v\otimes\psi)\big|^2  + \big|\partial_t (v\otimes\psi)\big|^2\\
&\qquad
+ \Big(m^2+H_2 -\dfrac{H_1^2}{4}+c\delta\Big) \big|(v\otimes\psi)\big|^2 \bigg]\dd s\dd t
+m\int_{\Sigma} \big|(v\otimes\psi)(s,0)\big|^2\dd s\\
&=\Big(\int_\Sigma \Big[(1+c\delta)|\nabla v|^2 +\Big(H_2-\dfrac{H_1^2}{4}+c\delta\Big) |v|^2\Big]\dd s\Big)\, \|\psi\|^2_{L^2(0,\delta)}\\
&\qquad+\Big( \int_{0}^\delta |\psi'|^2\dd t + m\big|\psi(0)\big|^2+m^2 \|\psi\|^2_{L^2(0,\delta)}\Big)\,\|v\|^2_{L^2(\Sigma,\CC^N)}\\
&=L_\delta[v,v] +\big(E_1(S) +m^2\big) \|v\|^2_\cH\le \big(E_j(L_\delta)+be^{-\delta|m|}\big)\|v\|^2_{\cH}.
\end{aligned}
\end{multline*}
Inserting the preceding inequalities into the expression \eqref{eq-bmm2} for $B^2_{m,M}$ one sees that
for all $u\in U$ there holds
\[
B_{m,M}^2[u,u]\le \Big(E_j(L_\delta)+be^{-\delta|m|}+a_0bC\dfrac{|m|}{M}\Big)\|v\|^2_{\cH}, \quad \|u\|^2_{L^2(\RR^n,\CC^N)}\ge \|v\|^2_\cH,
\]
and the min-max principle gives
\[
E_j(B_{m,M}^2)\le \max_{0\ne u\in U} \dfrac{B_{m,M}^2[u,u]}{\|u\|^2_{L^2(\RR^n,\CC^N)}}
=\max_{0\ne v\in V} \dfrac{B_{m,M}^2[u,u]}{\|u\|^2_{L^2(\RR^n,\CC^N)}}
\le E_j(L_\delta)+be^{-\delta|m|}+a_0bC\dfrac{|m|}{M}.
\]
Therefore, one has $\limsup_{m\to -\infty,\, m/M\to 0} E_j(B_{m,M}^2)\le E_j(L_\delta)$.
As $\delta$ can be chosen arbitrarily small and $\lim_{\delta\to 0} E_j(L_\delta)=E_j(L_0)\equiv E_j(L)$ one arrives at
\begin{equation}
  \label{eq-thm3a}
\limsup_{m\to -\infty,\, m/M\to 0} E_j(B_{m,M}^2)\le E_j(L).
\end{equation}

\subsection{Lower bound}
Now we will use the representation \eqref{eq-bmm1} with $\varepsilon=\varepsilon_0/|m|$ and an arbitrary but fixed $\varepsilon_0>0$, i.e.
\begin{multline*}
B_{m,M}^2[u, u]\\
=\int_{\Omega} \big(|\nabla u|^2 +m^2|u|^2\big)\dd x +\int_\Sigma \Big(m- \dfrac{\varepsilon_0}{|m|}-\dfrac{H_1}{2}\Big)|u|^2\dd s + 2(M-m)\int_\Sigma | \cP_- u|^2\dd s\\
+ \int_{\Omega^c} \big(|\nabla u|^2 +M^2|u|^2\big)\dd x -\int_\Sigma \Big(M - \dfrac{\varepsilon_0}{|m|}+\dfrac{H_1}{2}\Big)|u|^2\dd s, \quad u\in H^1(\RR^n,\CC^N).
\end{multline*}
Due to the min-max principle for any $j\in\NN$ one has
\begin{equation}
  \label{ebm01}
E_j(B_{m,M}^2)\ge E_j(K_{m,M}\oplus K^c_{m,M}),
\end{equation}
where
$K_{m,M}$ is the self-adjoint operator in $L^2(\Omega,\CC^N)$ with the form domain given by $\qdom(K_{m,M})=H^1(\Omega,\CC^N)$
and
\begin{multline*}
K_{m,M}[u,u]=\int_{\Omega} \big(|\nabla u|^2 +m^2|u|^2\big)\dd x\\
+\int_\Sigma \Big(m- \dfrac{\varepsilon_0}{|m|}-\dfrac{H_1}{2}\Big)|u|^2\dd s + 2(M-m)\int_\Sigma | \cP_- u|^2\dd s
\end{multline*}
and $K^c_{m,M}$ is the self-adjoint operator in $L^2(\Omega^c,\CC^N)$ with $\qdom(K^c_{m,M})=H^1(\Omega^c,\CC^N)$
and
\[
K^c_{m,M}[u,u]=\int_{\Omega^c} \big(|\nabla u|^2 +M^2|u|^2\big)\dd x -\int_\Sigma \Big(M - \dfrac{\varepsilon_0}{|m|}+\dfrac{H_1}{2}\Big)|u|^2\dd s.
\]
Using the operator $R_\gamma$ from Lemma~\ref{lem11} we see that in the asymptotic regime under consideration we have, with some $C_0>0$,
\begin{multline*}
E_1(K^c_{m,M})=E_1(R_{M-\varepsilon_0/m})+M^2\ge M^2-\Big( M- \dfrac{\varepsilon_0}{|m|}\Big)^2-C_0\\
=2\varepsilon_0 \dfrac{M}{|m|}-\dfrac{\varepsilon_0^2}{m^2}-C_0\to+\infty.
\end{multline*}
As we have already the upper bound $E_j(B_{m,M}^2)=\cO(1)$, it follows from \eqref{ebm01} that
$E_j(B_{m,M}^2)\ge E_j(K_{m,M})$. One can assume in addition that $M\ge 0$ and $m\le 0$,
then $2(M-m)\ge -2m\ge 2|m|$,
which implies
\begin{equation}
    \label{ebm02}
E_j(B_{m,M}^2)\ge E_j(K_m),
\end{equation}
with $K_m$ being the self-adjoint operator in $L^2(\Omega,\CC^N)$ with $\qdom(K_m)=H^1(\Omega,\CC^N)$
and
\[
K_m[u,u]=
\int_{\Omega} \big(|\nabla u|^2 +m^2|u|^2\big)\dd x +\int_\Sigma \Big(m- \dfrac{\varepsilon_0}{|m|}-\dfrac{H_1}{2}\Big)|u|^2\dd s + 2|m|\int_\Sigma | \cP_- u|^2\dd s.
\]
In order to obtain a lower bound for the eigenvalues of $K_m$ we take a small $\delta>0$
and consider the domains $\Omega_\delta=\big\{x\in\Omega: \dist(x,\delta)\big\}$ and $\Omega^c_\delta:=\Omega\setminus\overline{\Omega_\delta}$,
then due to the min-max principle one has
\begin{equation}
    \label{ekm03}
E_j(K_m)\ge E_j(K'_m\oplus K''_m),
\end{equation}
where $K'_m$ is the self-adjoint operator in $L^2(\Omega_\delta,\CC^N)$ with the form domain
$\qdom(K_m')=H^1(\Omega_\delta,\CC^N)$ and
\[
K'_m[u,u]=\int_{\Omega_\delta} \big(|\nabla u|^2 +m^2|u|^2\big)\dd x +\int_\Sigma \Big(m- \dfrac{\varepsilon_0}{|m|}-\dfrac{H_1}{2}\Big)|u|^2\dd s + |m|\int_\Sigma | \cP_- u|^2\dd s,
\]
while $K''_m$ is the self-adjoint operator in $L^2(\Omega_\delta^c,\CC^N)$ with 
\[
\qdom(K''_m)=H^1(\Omega^c_\delta,\CC^N),
\quad
K''_m[u,u]=\int_{\Omega^c_\delta} \big(|\nabla u|^2 +m^2|u|^2\big)\dd x,
\]
and $E_1(K''_m)\ge m^2\to +\infty$. By combining \eqref{ebm02} and~\eqref{ekm03} one sees that
$E_j(B_{m,M}^2)\ge E_j(K'_m\oplus K''_m)$. As we already have proved the upper bound
$E_j(B^2_{m,M})=\cO(1)$, it follows that
\begin{equation}
   \label{ebm04}
E_j(B_{m,M}^2)\ge E_j(K_m').
\end{equation}
Using now the diffeomorphism
\[
\Phi:\Pi_\delta\to \Omega_\delta,  \quad \Pi_\delta:=\Sigma\times(0,\delta),
\quad
\Phi(s,t)\mapsto s-t\nu(s),
\]
and the unitary maps
$\Theta_\delta: L^2(\Omega_\delta,\CC^N)\to L^2(\Pi_\delta,\CC^N)$, 
$\Theta_\delta u= \sqrt{\det (\Phi')}\,\, u\circ\Phi$,
with the help of Lemma~\ref{lem7}(b) one obtains
$E_j(K'_m)=E_j(\Theta_\delta^* K'_m\Theta_\delta)\ge E_j(K^0_m)$
with $K^0_m$ being the self-adjoint operator in $L^2(\Pi_\delta,\CC^N)$ given by
\begin{multline}
   \label{ekk0}
K^0_m[v,v]=\int_{\Pi_\delta} \bigg[(1-c\delta) |\nabla_s v|^2 + |\partial_t v|^2 + \Big(H_2 -\dfrac{H_1^2}{4}-c\delta\Big)
|v|^2 \bigg]\dd s\dd t\\
 +\Big(m-\dfrac{\varepsilon_0}{|m|}\Big)\int_\Sigma \big|v(s,0)\big|^2\dd s-c\int_\Sigma \big|v(s,\delta)\big|^2\dd s
+|m|\int_\Sigma \big|\cP_- v(s,0)\big|^2\dd s
\end{multline}
on the form domain $\qdom(K^0_m)=H^1(\Pi_\delta,\CC^N)$,
where $c>0$ is chosen independent of $\delta$ and $v$.
With this choice of  $c$, let $S'$ be the self-adjoint operator in $L^2(0,\delta)$ with
\[
S'[f,f]=\int_0^\delta |f'|^2\dd t +\Big(m-\dfrac{\varepsilon_0}{|m|}\Big) \big|f(0)\big|^2 - c \big|f(\delta)\big|^2,
\quad
\qdom (S')=H^1(0,\delta).
\]
and $\psi_k\in L^2(0,\delta)$ with $k\in\NN$ be its eigenfunctions for the eigenvalues $E_k(S')$ forming an orthonormal basis in $L^2(0,\delta)$. Due to Lemma~\ref{lem1dr} we have, with some $b^\pm>0$, $b>0$ and $b_0>0$,
\begin{gather}
    \label{conv05a}
E_1(S')\ge -\Big(|m|+\dfrac{\varepsilon_0}{|m|}\Big)^2-b e^{-\delta|m|}\ge -m^2-3\varepsilon_0 \text{ as } m\to-\infty, \\
    \label{conv06a}
b^- k^2- b_0\le E_k(S')\le b^+ k^2 \text{ for all $k\ge 2$ and $m\in \RR$.}
\end{gather}
For small $a\in\RR$, in addition to the operator $L_a$ in $\cH$
defined in \eqref{lavv1} we consider the self-adjoint operator $\Lambda_a$ in $L^2(\Sigma,\CC^N)$ given by
\[
\Lambda_a[g,g]=\int_\Sigma \Big[(1+ca)|\nabla g|^2 +\Big(H_2-\dfrac{H_1^2}{4}+ca\Big) |g|^2\Big]\dd s,
\quad
\qdom(\Lambda_a)=H^1(\Sigma,\CC^N).
\]

Let $K^1_m$ be the self-adjoint operator in $L^2(\Pi_\delta)$ having the same form domain
as $K^0_m$ and with the sesquilinear form obtained from the one of $K^0_m$ by omitting the last summand in \eqref{ekk0},
then $K^1_m$ admits a separation of variables: using the identification
$L^2(\Pi_\delta)\simeq L^2(0,\delta)\otimes L^2(\Sigma,\CC^N)$ one has
$K^1_m=S'\otimes 1+1\otimes\Lambda_{-\delta}$.
Using the unitary transform
\[
\Theta:L^2(0,\delta)\to \ell^2(\NN), \quad
(\Theta f)_k=\langle \psi_k,f\rangle_{L^2(0,\delta)}, \quad k\in\NN,
\]
the identification $L^2(\Pi_\delta)\simeq L^2(0,\delta)\otimes L^2(\Sigma,\CC^N)$
and another unitary transform
\begin{gather*}
\Xi:=\Theta\otimes 1:L^2(\Pi_\delta)\to \ell^2(\NN)\otimes L^2(\Sigma,\CC^N),\\
\Xi v=(v_k)=:\widehat v, \quad v_k:=\int_0^\delta \psi_k(t)\,v(t,\cdot)\dd t\in L^2(\Sigma,\CC^N),
\end{gather*}
for the self-adjoint operator $\widehat K^1_m:=\Xi K^1_m \Xi^*$ in $\ell^2(\NN)\otimes L^2(\Sigma,\CC^N)$ one has
\[
\widehat K^1_m[\widehat v,\widehat v]=
\sum_{k\in\NN} \Big( \Lambda_{-\delta}[v_k,v_k] +\big(E_k(S')+m^2\big)\|v_k\|^2_{L^2(\Sigma,\CC^N)},
\]
while $\qdom(K^1_m)$ consists of all $\widehat v\in \ell^2(\NN)\otimes L^2(\Sigma,\CC^N)$ with
$v_k\in H^1(\Sigma,\CC^N)$ such that the right-hand side of the preceding expression is finite.
Using the two-sided estimate \eqref{conv06a} one can rewrite
\begin{multline}
\qdom(\widehat K^1_m)=\Big\{ \widehat v=(v_k)\in \ell^2(\NN)\otimes L^2(\Sigma,\CC^N):\,
v_k\in H^1(\Sigma,\CC^N) \text{ for each $k\in \NN$}\\
\text{and } \sum_{k\in\NN} \Big(\|v_k\|^2_{H^1(\Sigma,\CC^N)}+k^2\|v_k\|^2_{H^1(\Sigma,\CC^N)}\Big)<\infty\Big\}.
   \label{qdomy0a}
\end{multline}
For the operator $\widehat K^0_m:=\Xi K^0_m \Xi^*$ one has the same form domain and
\[
\widehat K^0_m[\widehat v,\widehat v]=
\sum_{k\in\NN} \Big( \Lambda_{-\delta}[v_k,v_k] +\big(E_k(S')+m^2\big)\|v_k\|^2_{L^2(\Sigma,\CC^N)}\Big)
+|m|\int_\Sigma \big|\cP_- \Xi^* \widehat v(\cdot,0)\big|^2\dd s.
\]
Using the lower bounds \eqref{conv05a} and \eqref{conv06a} for $E_k(S')$, for any
$\widehat v\in \qdom(\widehat K^0_m)$
we obtain the inequality $\widehat K^0_m[\widehat v,\widehat v]\ge w_m(\widehat v,\widehat v)$
with the sesquilinear form $w_m$ in $\ell^2(\NN)\otimes L^2(\Sigma,\CC^N)$
defined on $\dom (w_m):=\qdom(\widehat K^0_m)$ by
\begin{multline*}
w_m(\widehat v,\widehat v):=\Lambda_{-\delta}[v_1,v_1]-3\varepsilon_0\|v_1\|^2_{L^2(\Sigma,\CC^N)}\\
+\sum_{k\ge 2} \Big( \Lambda_{-\delta}[v_k,v_k] + (b^-k^2-b_0 +m^2)\|v_k\|^2_{L^2(\Sigma,\CC^N)}\Big)
+
|m|\int_\Sigma \big|\cP_-\Xi^* \widehat v(\cdot,0)\big|^2\dd s.
\end{multline*}
Using the above representation \eqref{qdomy0a} one sees that the form $w_m$  is lower semibounded and closed, hence it generates a self-adjoint operator $W_m$
in $\ell^2(\NN)\otimes L^2(\Sigma,\CC^N)$ with compact resolvent, and then $E_j(\Hat K^0_m)\ge E_j(W_m)$ for all $j\in\NN$.
By summarizing all the preceding constructions, for any $j\in\NN$ in the asymptotic regime under consideration one has
\begin{equation}
   \label{conv07b}
E_j(B_{m,M}^2)\ge E_j(W_m).
\end{equation}
For the analysis of the eigenvalues of $W_m$ as $m\to-\infty$ we are now in the classical situation for the monotone convergence (Proposition~\ref{prop-mon}), as $W_m$
are increasing with respect to $|m|$. Namely, consider the set
\[
\cQ_\infty:=\Big\{\widehat v=(v_k)\in \bigcap_{m<0} \qdom (W_m)\equiv \qdom(\widehat K^0_m), \quad \sup_{m<0} W_m[\widehat v,\widehat v]<+\infty\Big\},
\]
then a vector $\widehat v=(v_k)\in \qdom(\widehat K^0_m)$ belongs to $\cQ_\infty$
iff the following two conditions are satisfied: (i) $v_k=0$ for all $k\ge 2$ and (ii) $\cP_- \Xi^* \widehat v(\cdot,0)=0$.
The condition (i) gives $v=e_1\otimes v_1$ with $e_1=(1,0,0,\dots)\in\ell^2(\NN)$, and then the condition
(ii) reduces to $\cP_- v_1=0$, i.e. $v_1\in\cH$.
Therefore,
\[
\cQ_\infty=\big\{e_1\otimes v_1: \, v_1\in H^1(\Sigma,\CC^N)\cap\cH\big\}.
\]
Moreover, for any $e_1\otimes v_1\in\cQ_\infty$ one has
\[
\lim_{m\to-\infty} W_m[e_1\otimes v_1,e_1\otimes v_1]=L_{-\delta}[v_1,v_1] -3\varepsilon_0\|v_1\|^2_{\cH},
\]
while we recall that $L_{-\delta}$ is defined as in \eqref{lavv1}.
Therefore, if one denotes by $W_\infty$ the self-adjoint operator in $e_1\otimes\cH$ given by
\[
W_\infty[e_1\otimes v_1,e_1\otimes v_1]=L_{-\delta}[v_1,v_1] -3\varepsilon_0\|v_1\|^2_{\cH},
\]
then it follows by the monotone convergence (Proposition~\ref{prop-mon})
that for each $j\in\NN$ there holds $\lim_{m\to-\infty}E_j(W_m)=E_j(W_\infty)\equiv E_j(L_{-\delta})-3\varepsilon_0$.
By \eqref{conv07b} one has
$\liminf_{M\to+\infty,m\to-\infty,m/M\to 0}\ge E_j(L_{-\delta})-3\varepsilon_0$.
As both $\delta$ and $\varepsilon_0$ can be chosen arbitrarily small and we have the convergence $\lim_{a\to0}E_j(L_a)=E_j(L)$,
we arrive at the inequality $\liminf_{M\to+\infty,m\to-\infty,m/M\to 0} E_j(B_{m,M}^2)\ge E_j(L_)$.
By combining it with the upper bound \eqref{eq-thm3a} we arrive at the result of Theorem~\ref{thm3}.

\appendix

\section{Schr\"odinger-Lichnerowicz formula for extrinsic Dirac operators on Euclidean hypersurfaces}\label{sec-lichn}

Let $\Sigma\subset\RR^n$ be a smooth compact hypersurface with the outer unit normal field $\nu$ and endowed with the Riemannian metric induced by the embedding.
Recall that the standard scalar product in $\RR^n$ gives  rise to the induced scalar product in $T\Sigma$, which we simply denote by $\langle\cdot,\cdot\rangle$ in this section.
Denote by $W$ the Weingarten operator, $W X=\nabla_X \nu$ for $X\in T\Sigma$, with $\nabla$ being the gradient in $\RR^n$.
Recall that the Levi-Civita connection $\nabla'$ on $\Sigma$
is given by the Gauss formula
\[
\nabla'_X Y=\nabla_X Y+\langle WX,Y\rangle\nu, \quad X,Y\in T\Sigma.
\]
We denote
\[
H_1:=\tr W, \quad |W|^2:=\tr (W^2), \quad H_2:=\dfrac{H_1^2-|W|^2}{2},
\]
i.e. $H_1$ is the mean curvature and $H_2$ is the half of the scalar curvature of $\Sigma$.

Let $N\in\NN$ and $\gamma_1,\dots,\gamma_n$ be $N\times N$ anticommuting Hermitian matrices satisfying
$\gamma_j^2=I$, with $I$ being the identity matrix, then the matrices
\[
\gamma(x):=\sum_{j=1}^n x_j \gamma_j, \quad x=(x_1,\dots,x_n)\in\RR^n,
\]
satisfy the commutation relation $\gamma(x)\gamma(y)+\gamma(x)\gamma(y)=2\langle x,y\rangle I_N$ for all $x,y\in\RR^n$.
Let us recall the definition of the associated extrinsic Dirac operator $D^\Sigma$
on $\Sigma$ following \cite{hmw}.
The induced spin connection $\nabla^\Sigma$ on $\Sigma$ is defined by
\[
\nabla^\Sigma_X \psi=\nabla_X+\dfrac{1}{2}\,\gamma(\nu)\gamma(WX):\, C^\infty(\Sigma,\CC^N)\to C^\infty(\Sigma,\CC^N), \quad
X\in T\Sigma,
\]
then $D^\Sigma$ acts on functions $\psi\in C^\infty(\Sigma,\CC^N)$ by
\[
D^\Sigma \psi:=-\gamma(\nu) \sum_{j=1}^{n-1} \gamma(e_j) \nabla^\Sigma_{e_j} \psi
\]
with $(e_1,\dots,e_{n-1})$ being an orthonormal frame of $T\Sigma$. Recall that
$\gamma(e_j)$ anticommute with $\gamma(\nu)$ and, furthermore,
\begin{equation}
    \label{eqwh1}
\sum_{j=1}^{n-1} \gamma(e_j) \gamma(We_j)=H_1\, I
\end{equation}
(which is seen by testing on an eigenbasis of $W$), and we may rewrite
\[
D^\Sigma \psi=\dfrac{H_1}{2}\,\psi-\gamma(\nu) \sum_{j=1}^{n-1} \gamma(e_j) \nabla_{e_j}\psi,
\]
Being viewed as an operator in $L^2(\Sigma,\CC^N)$, the operator $D^\Sigma$
is known to be essentially self-adjoint on $C^\infty(\Sigma,\RR^N)$.
We would like to provide a elementary direct proof, adapted to the Euclidean setting, of the eminent Schr\"odinger-Lichnerowicz formula
\begin{equation}
   \label{lichnapp}
(D^\Sigma)^2=(\nabla^\Sigma)^*\nabla^\Sigma+\dfrac{H_2}{2}\, I,
\end{equation}
where the first term on the right-hand side is the Bochner Laplacian associated with the above spin connection $\nabla^\Sigma$,
which is a self-adjoint operator in $L^2(\Sigma,\CC^N)$.
(We refer to the original papers \cite{Lic,Sch} and the monographs \cite{moroianu,fried,ginoux} for a more general setting.)

In what follows we use the standard identification of $T\Sigma$ and $T^*\Sigma$ with the help of the musical isomorphism.
Remark first that for $\psi\in C^\infty(\Sigma,\CC^N)$ we have the decomposition
\begin{equation}
    \label{nabla1}
\nabla^\Sigma \psi=\sum_{j=1}^{n-1} e_j\otimes \nabla^\Sigma_{e_j} \psi=
\sum_{j=1}^{n-1} e_j\otimes \Big(\nabla_{e_j}+\dfrac{1}{2}\,\gamma(\nu)\gamma(W e_j)\Big) \psi.
\end{equation}
Let us compute the adjoint $(\nabla^\Sigma)^*:T^*\Sigma\otimes C^\infty(\Sigma,\CC^N)\to C^\infty(\Sigma,\CC^N)$.
For $X\in T\Sigma\simeq T^* \Sigma$ and $\varphi,\psi\in C^\infty(\Sigma,\CC^N)$ we have
\begin{multline*}
\big\langle(\nabla^\Sigma)^*(X\otimes \varphi),\psi\big\rangle_{L^2(\Sigma,\CC^N)}=\langle X\otimes\varphi,\nabla^\Sigma\psi\rangle_{T^*\Sigma\otimes L^2(\Sigma,\CC^N)}\\
=\big\langle \varphi, \nabla_X \psi+\dfrac{1}{2}\,\gamma(\nu)\gamma(WX)\psi\big\rangle_{L^2(\Sigma,\CC^N)}
=\big\langle \varphi, \nabla_X \psi\big\rangle_{L^2(\Sigma,\CC^N)}\\
+\Big\langle \dfrac{1}{2}\,\gamma(WX) \gamma(\nu)\varphi,\psi\Big\rangle_{L^2(\Sigma,\CC^N)}.
\end{multline*}
Using Leibniz rule and the divergence theorem we have
\begin{align*}
\langle \varphi, \nabla_X \psi\rangle_{L^2(\Sigma,\CC^N)}&=
\int_\Sigma X\langle \varphi,\psi\rangle_{\CC^N} \dd s-\langle \nabla_X \varphi, \psi\rangle_{L^2(\Sigma,\CC^N)}\\
&=-\big\langle (\ddiv_\Sigma X) \varphi +\nabla_X \varphi,\psi\big\rangle_{L^2(\Sigma,\CC^N)},
\end{align*}
where $\ddiv_\Sigma$ is the divergence on $\Sigma$, 
\[
\ddiv_\Sigma X=\sum_{j=1}^{n-1} \langle e_j,\nabla'_{e_j} X\rangle.
\]
Therefore,
\[
(\nabla^\Sigma)^*(X\otimes \varphi)=-(\ddiv_\Sigma X)\,\varphi-\nabla_X \varphi + \frac{1}{2}\,\gamma(WX)\gamma(\nu)\varphi. \nonumber
\]
By combining \eqref{nabla1} with the last expression, for $\psi\in C^\infty (\Sigma,\CC^N)$ one obtains
\begin{equation}
  \label{eq-nab2}
\begin{aligned}
(\nabla^\Sigma)^* \nabla^\Sigma\psi&=\sum_{j=1}^{n-1} (\nabla^\Sigma)^* \Big[e_j\otimes \Big(\nabla_{e_j}+\dfrac{1}{2}\,\gamma(\nu)\gamma(W e_j)\Big) \psi\Big]\\
&=-\sum_{j=1}^{n-1} (\ddiv_\Sigma e_j) \Big(\nabla_{e_j}+\dfrac{1}{2}\,\gamma(\nu)\gamma(W e_j)\Big) \psi\\
&\quad +\sum_{j=1}^{n-1}\bigg\{
-\nabla_{e_j}\Big(\nabla_{e_j}+\dfrac{1}{2}\,\gamma(\nu)\gamma(W e_j)\Big)\psi\\
&\quad+\dfrac{1}{2}\,\gamma(We_j)\gamma(\nu)\Big(\nabla_{e_j}+\dfrac{1}{2}\,\gamma(\nu)\gamma(W e_j)\Big)\psi
\bigg\}=:S_1+S_2.
\end{aligned}
\end{equation}
To simplify $S_1$ we first use the Leibniz rule and the orthogonality of $(e_j)$ to obtain
\[
\ddiv_\Sigma e_j=\sum_{k=1}^{n-1} \langle e_k,\nabla'_{e_k} e_j\rangle=-\sum_{k=1}^{n-1} \langle \nabla'_{e_k} e_k, e_j\rangle
\]
and 
\begin{align*}
S_1&=\sum_{j,k=1}^{n-1} \langle \nabla'_{e_k} e_k, e_j\rangle\nabla_{e_j}\psi
+\dfrac{1}{2}\sum_{j,k=1}^{n-1} \langle \nabla'_{e_k} e_k, e_j\rangle \gamma(\nu)\gamma(W e_j) \psi\\
&=\sum_{k=1}^{n-1} \Big(\sum_{j=1}^{n-1}\langle \nabla'_{e_k} e_k, e_j\rangle \nabla_{e_j}\psi\Big)\\
&\qquad +\dfrac{1}{2}\sum_{k=1}^{n-1} \gamma(\nu) \gamma\bigg( W\sum_{j=1}^{n-1}\langle \nabla'_{e_k} e_k, e_j\rangle e_j \bigg)\\
&=\sum_{k=1}^{n-1} \nabla_{\nabla'_{e_k} e_k} \psi +\dfrac{1}{2}\sum_{k=1}^{n-1} \gamma(\nu)\gamma \Big( W \nabla'_{e_k} e_k\Big).
\end{align*}
Furthermore,
\begin{align*}
S_2&=\sum_{j=1}^{n-1} \bigg\{-\nabla_{e_j}\nabla_{e_j}\psi-\dfrac{1}{2}\,\gamma(We_j)\gamma(W e_j)\psi
-\dfrac{1}{2}\,\gamma(\nu)\gamma\big(\nabla_{e_j}(W e_j)\big)\psi\\
&\qquad -\dfrac{1}{2}\,\gamma(\nu)\gamma(We_j)\nabla_{e_j}\psi+\dfrac{1}{2}\,\gamma(We_j)\gamma(\nu)\nabla_{e_j}\psi+ \dfrac{1}{4}\,\gamma(We_j)\gamma(\nu)\gamma(\nu)\gamma(W e_j)\psi\bigg\}\\
&=\sum_{j=1}^{n-1} \bigg\{-\nabla_{e_j}\nabla_{e_j}\psi-\dfrac{1}{2}\,\gamma(\nu)\gamma\big(\nabla_{e_j}(W e_j)\big)\psi-\gamma(\nu)\gamma(We_j)\nabla_{e_j}\psi\bigg\}\psi
- \dfrac{1}{4}\,|W|^2\psi,
\end{align*}
and then
\begin{multline*}
(\nabla^\Sigma)^* \nabla^\Sigma\psi=
\sum_{j=1}^{n-1} \bigg[ \nabla_{\nabla'_{e_j} e_j} \psi
-\nabla_{e_j}\nabla_{e_j}\psi\\
+\dfrac{1}{2}\gamma(\nu)\gamma \Big( W \nabla'_{e_j} e_j-\nabla_{e_j}(W e_j)\Big)\psi
-\gamma(\nu)\gamma(We_j)\nabla_{e_j}\psi\bigg]
- \dfrac{1}{4}\,|W|^2\psi.
\end{multline*}
Using $\nabla'_{e_j} (W e_j)=\nabla_{e_j} (W e_j)+|W e_j|^2\nu$ and Leibniz rule we have
\[
W \nabla'_{e_j} e_j-\nabla_{e_j}(W e_j)=W\nabla'_{e_j} e_j-\nabla'_{e_j} (W e_j)+|W e_j|^2\nu
=-(\nabla'_{e_j} W)e_j +|W e_j|^2\nu
\]
implying $\gamma(\nu)\gamma \big( W \nabla'_{e_j} e_j-\nabla_{e_j}(W e_j)\big)\psi=-\gamma(\nu)\gamma\big((\nabla'_{e_j} W)e_j\big)\psi+|W e_j|^2\psi$,
and then
\begin{multline}
    \label{lapl1}
(\nabla^\Sigma)^* \nabla^\Sigma\psi=
\sum_{j=1}^{n-1} \bigg[ \nabla_{\nabla'_{e_j} e_j} \psi
-\nabla_{e_j}\nabla_{e_j}\psi\\
-\dfrac{1}{2} \,\gamma(\nu)\,\gamma\big((\nabla'_{e_j} W)e_j\big)\psi
-\gamma(\nu)\,\gamma(We_j)\nabla_{e_j}\psi\bigg]
+ \dfrac{1}{4}\,|W|^2\psi.
\end{multline}
On the other hand,
\begin{equation}
   \label{eq-dd0}
\begin{aligned}
D^\Sigma D^\Sigma\psi&=\Big(\dfrac{H_1}{2}-\gamma(\nu) \sum_{j=1}^{n-1} \gamma(e_j) \nabla_{e_j}\Big)\Big(\dfrac{H_1\,\psi}{2}-\gamma(\nu) \sum_{j=1}^{n-1} \gamma(e_j) \nabla_{e_j}\psi\Big)\\
&=\dfrac{H_1^2}{4}-\dfrac{1}{2} \Big( \gamma(\nu) \sum_{j=1}^{n-1} \gamma(e_j) \nabla_{e_j} H_1\Big)\, \psi - \dfrac{H_1}{2}\,\gamma(\nu) \sum_{j=1}^{n-1} \gamma(e_j) \nabla_{e_j}\psi\\
&\quad-\dfrac{H_1}{2}\,\gamma(\nu) \sum_{j=1}^{n-1} \gamma(e_j) \nabla_{e_j}\psi+\gamma(\nu)\sum_{j,k=1}^{n-1} \gamma(e_j)\gamma(We_j)\gamma(e_k)\nabla_{e_k}\psi\\
&\quad + \gamma(\nu)\sum_{j,k=1}^{n-1} \gamma(e_j)\gamma(\nu)\gamma(\nabla_{e_j}e_k)\nabla_{e_k}\psi\\
&\quad+\gamma(\nu)\sum_{j,k=1}^{n-1}\gamma(e_j)\gamma(\nu)\gamma(e_k)\nabla_{e_j}\nabla_{e_k}\psi.
\end{aligned}
\end{equation}
The sum of the second, third and forth terms is zero, in fact, 
\begin{align*}
&-\dfrac{H_1}{2}\,\gamma(\nu) \sum_{j=1}^{n-1} \gamma(e_j) \nabla_{e_j}\psi
-\dfrac{H_1}{2}\,\gamma(\nu) \sum_{j=1}^{n-1} \gamma(e_j) \nabla_{e_j}\psi\\
&\quad+\gamma(\nu)\sum_{j,k=1}^{n-1} \gamma(e_j)\gamma(We_j)\gamma(e_k)\nabla_{e_k}\psi\\
&=
-H_1\,\gamma(\nu) \sum_{j=1}^{n-1} \gamma(e_k) \nabla_{e_k}\psi
+\gamma(\nu)\sum_{j,k=1}^{n-1} \gamma(e_j)\gamma(We_j)\gamma(e_k)\nabla_{e_k}\psi\\
&=\gamma(\nu)\sum_{k=1}^{n-1} \Big( -H_1 +\sum_{j=1}^{n-1}\gamma(e_j)\gamma(We_j)\Big)\gamma(e_k)\nabla_{e_k}\psi=0
\end{align*}
as the term in the parentheses identically vanishes due to \eqref{eqwh1}. Therefore, Eq.~\eqref{eq-dd0} rewrites as
\begin{equation}
   \label{eq-dd1}
\begin{aligned}
(D^\Sigma)^2\psi&=\dfrac{H_1^2}{4}-\dfrac{1}{2} \, \gamma(\nu) \sum_{j=1}^{n-1} \gamma(e_j) (\nabla_{e_j} H_1)\, \psi
- \sum_{j,k=1}^{n-1} \gamma(e_j)\gamma(\nabla_{e_j}e_k)\nabla_{e_k}\psi\\
&\quad -\sum_{j,k=1}^{n-1}\gamma(e_j)\gamma(e_k)\nabla_{e_j}\nabla_{e_k}\psi.
\end{aligned}
\end{equation}
We transform the last term in this expression as follows:
\begin{multline*}
\sum_{j,k=1}^{n-1}\gamma(e_j)\gamma(e_k)\nabla_{e_j}\nabla_{e_k}\psi\\
\begin{aligned}
=\,&\dfrac{1}{2}\,\sum_{j,k=1}^{n-1}\Big(\gamma(e_j)\gamma(e_k)\nabla_{e_j}\nabla_{e_k}\psi+\gamma(e_k)\gamma(e_j)\nabla_{e_k}\nabla_{e_j}\psi\Big)\\
=\, &\dfrac{1}{2}\,\sum_{j,k=1}^{n-1}\Big(\gamma(e_j)\gamma(e_k)+\gamma(e_k)\gamma(e_j)\Big)\nabla_{e_j}\nabla_{e_k}\psi\\
&+\dfrac{1}{2}\sum_{j,k=1}^{n-1}\gamma(e_k)\gamma(e_j) \Big( \nabla_{e_k}\nabla_{e_j}-\nabla_{e_j}\nabla_{e_k}\Big)\psi\\
=\,&\sum_{j=1}^{n-1}\nabla_{e_j}\nabla_{e_j}\psi+\dfrac{1}{2}\,J,
\end{aligned}
\end{multline*}
where
\[
J:=\sum_{j,k=1}^{n-1}\gamma(e_j)\gamma(e_k) \Big( \nabla_{e_j}\nabla_{e_k}-\nabla_{e_k}\nabla_{e_j}\Big)\psi
\equiv
\sum_{j,k=1}^{n-1}\gamma(e_j)\gamma(e_k) \nabla_{[e_j,e_k]}\psi.
\]
Representing $[e_j,e_k]=\sum_{l=1}^{n-1} \big\langle e_l,[e_j,e_k]\big\rangle e_k$
we have
\[
J=\sum_{j,k,l=1}^{n-1} \gamma(e_j)\gamma(e_k)\Big[\langle e_l,\nabla'_{e_j}e_k\rangle-\langle e_l,\nabla'_{e_k}e_j\rangle\Big]\nabla_{e_l}\psi,
\]
and using
\[
\sum_{j=1}^{n-1} e_j\langle e_l,\nabla'_{e_k}e_j\rangle=-\sum_{j=1}^{n-1} e_j\langle \nabla'_{e_k} e_l,e_j\rangle=-\nabla'_{e_k} e_l
\]
we rewrite
\begin{align*}
J&=-\sum_{j,k=1}^{n-1}\gamma(e_j)\,\gamma(\nabla'_{e_j}e_k) \nabla_{e_k}\psi+\sum_{j,k=1}^{n-1}\gamma(\nabla'_{e_j}e_k)\,\gamma(e_j)\, \nabla_{e_k}\psi\\
&=-2\sum_{j,k=1}^{n-1}\gamma(e_j)\,\gamma(\nabla'_{e_j}e_k) \nabla_{e_k}\psi\\
&\qquad+\sum_{j,k=1}^{n-1}\Big( \gamma(e_j)\,\gamma(\nabla'_{e_j}e_k)+\gamma(\nabla'_{e_j}e_k)\,\gamma(e_j)\,\Big) \nabla_{e_k}\psi\\
&=-2\sum_{j,k=1}^{n-1}\gamma(e_j)\,\gamma(\nabla'_{e_j}e_k) \nabla_{e_k}\psi+2\sum_{j,k=1}^{n-1}\langle e_j,\nabla'_{e_j}e_k\rangle \nabla_{e_k}\psi\\
&=-2\sum_{j,k=1}^{n-1}\gamma(e_j)\,\gamma(\nabla'_{e_j}e_k) \nabla_{e_k}\psi-2\sum_{j,k=1}^{n-1}\langle \nabla'_{e_j}e_j,e_k\rangle \nabla_{e_k}\psi\\
&=-2\sum_{j,k=1}^{n-1}\gamma(e_j)\,\gamma(\nabla'_{e_j}e_k) \nabla_{e_k}\psi-2\sum_{j=1}^{n-1} \nabla_{\nabla'_{e_j}e_j}\psi.
\end{align*}
The substitution into \eqref{eq-dd1} gives
\begin{equation*}
\begin{aligned}
D^2_\Sigma \psi&=\dfrac{H_1^2}{4}-\dfrac{1}{2} \, \gamma(\nu) \sum_{j=1}^{n-1} \gamma(e_j) (\nabla_{e_j} H_1)\, \psi
- \sum_{j,k=1}^{n-1} \gamma(e_j)\gamma(\nabla_{e_j}e_k)\nabla_{e_k}\psi\\
&\quad -\sum_{j=1}^{n-1}\nabla_{e_j}\nabla_{e_j}\psi+\sum_{j,k=1}^{n-1}\gamma(e_j)\,\gamma(\nabla'_{e_j}e_k) \nabla_{e_k}\psi+\sum_{j=1}^{n-1} \nabla_{\nabla'_{e_j}e_j}\psi.
\end{aligned}
\end{equation*}
The sum of the third and the fifth terms simplifies as
\begin{multline*}
- \sum_{j,k=1}^{n-1} \gamma(e_j)\gamma(\nabla_{e_j}e_k)\nabla_{e_k}\psi+\sum_{j,k=1}^{n-1}\gamma(e_j)\,\gamma(\nabla'_{e_j}e_k) \nabla_{e_k}\psi\\
=\sum_{j,k=1}^{n-1} \gamma(e_j)\gamma(\nabla'_{e_j}e_k-\nabla_{e_j}e_k) \nabla_{e_k}\psi
=\sum_{j,k=1}^{n-1} \gamma(e_j)\gamma\big(\langle W e_j,e_k\rangle\nu\big) \nabla_{e_k}\psi\\
=\sum_{j,k=1}^{n-1} \gamma\big(e_j \langle  e_j, W e_k\rangle\big)\gamma(\nu) \nabla_{e_k}\psi=
\sum_{k=1}^{n-1}\gamma(W e_k) \gamma(\nu)\nabla_{e_k}\psi,
\end{multline*}
hence,
\begin{multline*}
D^2_\Sigma \psi=\dfrac{H_1^2}{4}-\dfrac{1}{2} \, \gamma(\nu) \sum_{j=1}^{n-1} \gamma(e_j) (\nabla_{e_j} H_1)\, \psi\\
+\sum_{j=1}^{n-1}\gamma(W e_j) \gamma(\nu)\nabla_{e_j}\psi
-\sum_{j=1}^{n-1}\nabla_{e_j}\nabla_{e_j}\psi+\sum_{j=1}^{n-1} \nabla_{\nabla'_{e_j}e_j}\psi.
\end{multline*}
By comparing the last expression with \eqref{lapl1} we obtain
\begin{multline*}
D^2_\Sigma \psi - (\nabla^\Sigma)^*\nabla^\Sigma\psi\\
=
\dfrac{H_1^2}{4}-\dfrac{1}{4} |W|^2 -\dfrac{1}{2} \, \gamma(\nu) \sum_{j=1}^{n-1} \gamma(e_j) (\nabla_{e_j} H_1)\, \psi
+\sum_{j=1}^{n-1}\gamma(W e_j) \gamma(\nu)\nabla_{e_j}\psi\\
+\dfrac{1}{2} \sum_{j=1}^{n-1}\,\gamma(\nu)\,\gamma\big((\nabla'_{e_j} W)e_j\big)\psi
+\sum_{j=1}^{n-1}\gamma(\nu)\,\gamma(We_j)\nabla_{e_j}\psi.
\end{multline*}
Noting that the sum of the fouth term and the sixth term on the right hand is zero,
we arrive at
\begin{multline*}
(D^\Sigma)^2 \psi - (\nabla^\Sigma)^*\nabla^\Sigma\psi\\
= \dfrac{H_2}{2}\,\psi-\dfrac{1}{2} \, \gamma(\nu) \sum_{j=1}^{n-1} \gamma(e_j) (\nabla_{e_j} H_1)\, \psi
+\dfrac{1}{2} \sum_{j=1}^{n-1}\,\gamma(\nu)\,\gamma\big((\nabla'_{e_j} W)e_j\big)\psi\\
=\dfrac{H_2}{2}\,\psi + \dfrac{1}{2}\gamma(\nu) \gamma \bigg(\sum_{j=1}^{n-1} (\nabla'_{e_j} W) e_j - \sum_{j=1}^{n-1} (\nabla_{e_j} H_1)e_j\bigg)\psi.
\end{multline*}
Therefore, in order to show the sought identity \eqref{lichnapp} it is sufficient to prove the  equality
\begin{equation}
  \label{ric0}
\sum_{j=1}^{n-1} (\nabla'_{e_j} W)e_j=\sum_{j=1}^{n-1} (\nabla_{e_j} H_1)e_j.
\end{equation}
In order to check \eqref{ric0} let us remark that $\nabla_X \nabla_Y Z-\nabla_Y \nabla_Z - \nabla_{[X,Y]} Z=0$ for any $X,Y,Z\in T\Sigma$.
Using the definition of $\nabla'$ we have
\begin{align*}
0=\,&\nabla_X \big( \nabla'_Y Z - \langle WY,Z\rangle\nu\big)-\nabla_Y\big(\nabla'_X Z -\langle WX,Z\rangle \nu\big)- \nabla'_{[X,Y]} Z +\big\langle W[X,Y],Z\big\rangle\nu\\
=\,&\nabla'_X\big( \nabla'_Y Z - \langle WY,Z\rangle\nu\big)- \big\langle WX, \nabla'_Y Z - \langle WY,Z\rangle\nu\big\rangle
- \nabla'_Y\big(\nabla'_X Z -\langle WX,Z\rangle \nu\big)\\
& +\big\langle WY, \nabla'_X Z -\langle WX,Z\rangle \nu\big\rangle- \nabla'_{[X,Y]} Z +\big\langle W[X,Y],Z\big\rangle\nu.
\end{align*}
Using $\nabla'_X \nu=\nabla_X \nu=WX$ we then arrive at
\begin{align*}
0=\,&\nabla'_X\nabla'_Y Z
- \big\langle (\nabla'_X W)Y,Z\big\rangle\nu-\big\langle W(\nabla'_X Y),Z\big\rangle\nu-\langle WY,\nabla'_X Z\rangle\nu-\langle WY,Z\rangle WX\\
&- \langle WX, \nabla'_Y Z\rangle\nu-\nabla'_Y \nabla'_X Z+ \big\langle (\nabla'_Y W)X,Z\big\rangle \nu+\langle W\nabla'_Y X,Z\rangle \nu\\
&+\langle WX,\nabla'_Y Z\rangle \nu+\langle WX,Z\rangle WY+\langle WY, \nabla'_X Z\rangle
- \nabla'_{[X,Y]} Z +\big\langle W[X,Y],Z\big\rangle\nu\\
=\,& \nabla'_X\nabla'_Y Z-\nabla'_Y \nabla'_X Z- \nabla'_{[X,Y]} Z
-\langle WY,Z\rangle WX +\langle WX,Z\rangle WY\\
& +\big\langle (\nabla'_Y W)X,Z\big\rangle \nu - \big\langle (\nabla'_X W)Y,Z\big\rangle\nu\\
& +\langle W\nabla'_Y X,Z\rangle \nu-\big\langle W(\nabla'_X Y),Z\big\rangle\nu +\big\langle W[X,Y],Z\big\rangle\nu.
\end{align*}
As $\nabla'_X Y-\nabla'_Y X=[X,Y]$, the sum of the terms in the last line vanishes, and considering the normal components of the remaining equality
we obtain $\big\langle (\nabla'_Y W)X,Z\big\rangle = \big\langle (\nabla'_X W)Y,Z\big\rangle$, and then
$\big\langle (\nabla'_Y W)X,Z\big\rangle=\big\langle Y,(\nabla'_X W) Z\big\rangle$.
Taking $Y=Z=e_k$ and summing over $k$ we arrive at
\[
\sum_{k=1}^{n-1}\big\langle (\nabla'_{e_k} W) X,e_k\big\rangle=\sum_{j=1}^{n-1} \big\langle e_k,(\nabla'_X W ) e_k\big\rangle
\text{ i.e. }  \sum_{k=1}^{n-1}\big\langle X, (\nabla'_{e_k} W) e_k\big\rangle=\nabla_X H_1.
\]
Using the last equality for $X=e_j$ we obtain
\[
\sum_{j=1}^{n-1}\sum_{k=1}^{n-1}\big\langle e_j, (\nabla'_{e_k} W) e_k\big\rangle e_j=\sum_{j=1}^{n-1} (\nabla_{e_j} H_1) e_j,
\]
and the left-hand side simplifies to $\sum_{k=1}^{n-1} (\nabla'_{e_k} W) e_k$, which gives \eqref{ric0} and finishes the proof of~\eqref{lichnapp}.

\section*{Acknowledgments}

The authors thank Christian G\'erard for numerous useful discussions. A large part of this paper was written while Thomas Ourmi\`eres-Bonafos was supported by a public grant as part of the ``Investissement d'avenir'' project, reference ANR-11-LABX-0056-LMH, LabEx LMH, at the University Paris-Sud.
Now, Thomas Ourmi\`eres-Bonafos is supported by the ANR "D\'efi des autres savoirs (DS10) 2017" programm, reference ANR-17-CE29-0004, project molQED.
Konstantin Pankrashkin was in part supported by the PRC 1556 CNRS-RFBR 2017--2019 ``Multi-dimensional semi-classical problems of condensed matter physics and quantum mechanics''.

\end{document}